\documentclass[12pt,english]{article}
\usepackage{amsfonts}
\usepackage[T1]{fontenc}
\usepackage{geometry}
\usepackage{color}
\usepackage{babel}
\usepackage{mathrsfs}
\usepackage{amsmath}
\usepackage{amssymb}
\usepackage{bm}
\usepackage{graphicx}
\usepackage{setspace}
\usepackage{esint}
\usepackage[
style=authoryear,
sorting=nyt,
sortcites=true,
maxcitenames=3,
mincitenames=1,
uniquelist=false,
maxbibnames=9,
backend=biber,
natbib,
citestyle=authoryear
]{biblatex}

\usepackage[unicode=true,pdfusetitle,
 bookmarks=true,bookmarksnumbered=true,bookmarksopen=false,
 breaklinks=false,pdfborder={0 0 0},backref=false,colorlinks=true]{hyperref}
\usepackage{aecompl}
\usepackage{geometry}
\usepackage{babel}
\usepackage{breakurl}
\usepackage{amsthm}
\usepackage{breakurl}
\usepackage{enumitem}
\setlist{
    topsep=4pt,
    partopsep=0pt,
    itemsep=2pt,
    parsep=0pt
}
\usepackage{rotating}
\usepackage{tocloft}
\usepackage{xcolor}
\usepackage{bigints}
\usepackage[nottoc]{tocbibind}
\usepackage{bookmark}
\usepackage{lscape}
\usepackage{cleveref}
\usepackage{tikz}
\usepackage{pgfplots}
\usepackage{subcaption}
\usepackage{caption}
\usepackage{gnuplottex}
\usepackage{lmodern}
\usepackage[inline]{asymptote}
\usepackage{appendix}
\usepackage{float}

\hypersetup{
 linkcolor=darkred, citecolor=darkblue}
\makeatletter
\colorlet{darkred}{red!70!black}
\colorlet{darkblue}{blue!70!black}
\numberwithin{equation}{section}
\numberwithin{figure}{section}
\theoremstyle{plain}
\newtheorem{theorem}{Theorem}

\newtheorem{corollary}{Corollary}

\newtheorem{lemma}{\protect\lemmaname}
\newtheorem*{lemma*}{Lemma}

\newtheorem{definition}{Definition}
\theoremstyle{definition}
\newtheorem{remark}{Remark}
\newtheorem{example}{Example}
\usepackage{amsthm}
\usepackage{dcolumn}
\usepackage{etoc}
\usepackage{titlesec}

\newtheorem{assumption*}[]{Assumption}

\newdimen\dummy
\dummy=\oddsidemargin
\allowdisplaybreaks

  \providecommand{\lemmaname}{Lemma}

\providecommand{\propositionname}{Proposition}

\DeclareMathOperator*{\supp}{supp}

\DeclareMathOperator*{\id}{id}
\makeatother

\newcommand{\setdisplaygaps}{%
    \setlength{\abovedisplayskip}{8pt plus 2pt minus 2pt}%
    \setlength{\belowdisplayskip}{8pt plus 2pt minus 2pt}%
    \setlength{\abovedisplayshortskip}{4pt plus 2pt minus 1pt}%
    \setlength{\belowdisplayshortskip}{4pt plus 2pt minus 1pt}%
}

\AtBeginDocument{\setdisplaygaps}
\AddToHook{cmd/normalsize/after}{\setdisplaygaps}
\AddToHook{cmd/small/after}{\setdisplaygaps}
\AddToHook{cmd/footnotesize/after}{\setdisplaygaps}

\usepackage{amsthm}
\usepackage{etoolbox}

\title{Stable Mechanisms\footnote{An earlier version of this paper (arXiv:2603.12532) is circulated under the title ``Self-Confirming Mechanisms''. We thank Jiangtao Li, Eran Shmaya, Tianhao Liu, Bruno Strulovici, Marina Halac, Navin Kartik, Eva Tardos, Maxwell Rosenthal, Leeat Yariv, Fedor Sandomirskiy, Sylvain Chassang, Ryota Iijima, Alessandro Pavan, and participants at various seminars and conferences for helpful comments and questions.}}

\author{Zhiming Feng\footnote{Department of Economics, Columbia University, zf2295@columbia.edu.}
\quad\quad
Qingmin Liu\footnote{Department of Economics, Columbia University, ql2177@columbia.edu.}}

\date{November 17, 2026}

\begin{document}

\maketitle

\maketitle
\begin{abstract}
We study mechanism design environments in which the designer does not know the distribution of agents’ private information and draws inferences from the behavior induced by the mechanism itself. A mechanism is stable if it is optimal under some belief consistent with the information it generates, and robustly stable if this property survives small informational perturbations. We establish a fictitious revelation principle: any incentive-compatible mechanism admits a representation as a direct mechanism paired with an information filter, preserving both outcomes and informational content. Applying the framework to the classic monopoly problem, we show that, under regularity of the true type distribution, robustly stable mechanisms coincide with optimal mechanisms under that distribution, as if it were known to the designer. In general, however, posted-price mechanisms can remain robustly stable at local revenue maxima, so globally suboptimal prices can persist.

\bigskip


\end{abstract}

\pagenumbering{roman}

\newpage
{
\singlespacing
\tableofcontents
}

\newpage

\pagenumbering{arabic}
\setcounter{page}{1}

\section{Introduction}

Mechanisms and institutions are not always the products of purely optimal normative design. Rather, they often emerge as equilibrium or steady-state outcomes of economic interaction. This paper develops a theory of mechanisms built around one key force of equilibration: the operation of a mechanism generates data through the responses of the agents who participate in it; these data shape the designer’s beliefs about the underlying environment and can, in turn, create incentives for redesign. Some mechanisms are destabilized by this informational feedback; others persist, potentially becoming trapped in seemingly suboptimal states. We develop a criterion of stability: a mechanism is stable precisely when the information generated by its own operation continues to justify its use.

A particularly natural environment in which these forces arise is one where neither the designer nor the economic agents know, a priori, the distribution of the underlying uncertainty. This concern echoes a long, though somewhat neglected, tradition in economics emphasizing the limits of centralized design when relevant knowledge is dispersed \citep{hayek1945use}. For example, a seller may face buyers who know their own valuations, while neither side has precise knowledge of the underlying distribution of valuations. Similarly, in the design of political institutions---such as voting rules, legislative procedures, or other mechanisms for aggregating public input---the distribution of citizens' preferences may be unknown to the designer and only imperfectly understood by citizens themselves. In such settings, uncertainty on the agents’ side is onerous but often not fatal, at least in theory: by focusing on dominant-strategy or ex post incentive-compatible mechanisms, agents’ behavior can be rendered independent of beliefs about others’ type distributions. By contrast, uncertainty on the designer’s side---about the distribution of agents’ private information or preferences---is critical in mechanism design.\footnote{Our claim concerns the question of how agents interact in a mechanism when they do not know others' type distributions, and how the designer evaluates such mechanisms. \citet{chung2007foundations}, by contrast, take a different perspective, asking whether and when it makes sense for a revenue-maximizing auctioneer to impose dominant-strategy incentive compatibility in the first place.}

One response to this problem in computer science and economic theory is robust or prior-free design, which departs from the classical paradigm by replacing optimization under a precisely specified prior with performance guarantees over a class of environments, often evaluated using worst-case criteria; see, e.g., \citet{chassang2013calibrated} and \citet{carroll2015robustness}.\footnote{A foundational antecedent is \citet{wilson1987game}'s call to weaken common-knowledge assumptions in game-theoretic modeling.} This approach is natural and has generated deep insights into normative mechanism design. However, there are two reasons to look beyond it. First, worst-case criteria embody an extreme form of pessimism, and the resulting guarantees can be very conservative, sometimes driven by highly contrived worst-case scenarios.\footnote{This concern about worst-case analysis is well recognized beyond mechanism design. In theoretical computer science, for example, it has motivated alternative approaches such as the smoothed analysis of \citet{spielman2004smoothed}.} Second, and more fundamentally, a mechanism that is robust in this sense may fail a consistency test: does the information generated by its own operation continue to justify its use? A mechanism designed against a worst-case environment may itself generate data that rule out the very environment used to justify it. In that case, the mechanism is unstable and vulnerable to redesign. The goal of this paper is to develop such a test.
Rather than specifying an explicit dynamic process, we propose a new reduced-form approach that captures an essential aspect of this feedback loop and the possibility of redesign. We view mechanism design as a fixed-point problem: mechanisms generate information through their operation, and this information in turn constrains which mechanisms can remain optimal. A mechanism that passes this test will be called ``stable.''

While the underlying idea is natural, formalizing it coherently requires addressing several conceptual issues and making explicit choices about the appropriate modeling and equilibrium postulates.
The first issue concerns what the designer learns from running a mechanism. Fix a mechanism that maps message profiles into outcomes, and fix a profile of agents’ strategies that maps types into possibly random messages. Together with the true underlying distribution over types, the strategy profile induces a distribution over message profiles. Recall that the central aim of this paper is to formulate a stability test: if a mechanism passes this test, the information generated by its operation does not give the designer reason to change it. If such a mechanism persists over time, the designer will eventually learn the long-run distribution of messages. Consistent with our reduced-form equilibrium perspective, we therefore take this distribution as observed, abstracting from the underlying dynamic learning process. When the designer observes this distribution and, as part of the equilibrium postulate, knows the agents’ strategies, she can infer only that the true distribution of types belongs to the set of distributions that generate the same message distribution under those strategies. This set captures exactly the identifying content of the observed data. Without further assumptions, no additional inference is possible. Augmented mechanisms can therefore be naturally compared by the informativeness of the feasible-prior sets they induce.

The second issue concerns the universal set of mechanisms.
As in classical mechanism design, we allow flexible
message spaces and outcome rules, together with incentive
compatible strategy profiles. The question is to
characterize and analyze this unrestricted class; this is why the mechanism space cannot be treated  
simply as an abstract policy space. The revelation principle of \citet{myerson1981optimal}
makes this task tractable through truthful direct
representations. Unlike in the classical framework, however, messages here also serve as a source of information from which the designer draws inferences about the distribution of agents’ types. This distinction has an immediate implication: the standard revelation principle is no longer sufficient for our purposes. Direct mechanisms with truth-telling strategies reveal the true distribution of types and therefore fail to capture mechanisms that pool information; for example, mechanisms with degenerate message spaces or strategies designed to conceal information. To study all stable mechanisms, we therefore need a revelation principle that preserves not only outcomes, but also informational content.

Before turning to our third conceptual point, we allow ourselves
a brief digression to clarify two issues that may be obvious to
some readers. One may ask why indirect mechanisms deserve attention.
A direct mechanism is often better understood as a theoretical
representation,
while partially revealing mechanisms are ubiquitous in practice.
A posted-price mechanism, for example, reveals only whether a buyer's
valuation exceeds the posted price. It is therefore worthwhile to understand whether and why such mechanisms persist. A related question is whether one could simply use a direct mechanism
to learn the true type distribution over time and then implement the
mechanism optimal for that distribution, thereby dispensing with
indirect mechanisms. This question takes the dynamic interpretation
seriously. However, modeling dynamic learning requires specifying how
ambiguity about type distributions is updated from finite samples,
whether such updating is dynamically consistent, and how current
payoffs are traded off against future information, among other
things.\footnote{One could instead impose a well-defined prior over
possible type distributions and proceed in a Bayesian fashion, but
doing so would reintroduce precisely the prior information that the
designer may not possess. See also
\citet{MorgensternRoughgarden2016} and
\citet{kanoria2021incentive} for related considerations and
references therein.} The advantage of the reduced-form approach is precisely that it abstracts from these dynamic considerations, which require additional modeling choices, and instead asks which mechanisms pass a consistency test: whether the information generated by a mechanism undermines or justifies its use.

The third conceptual issue concerns beliefs in a stable mechanism. The information endogenously generated by a mechanism determines a set of type distributions consistent with the empirical distribution of messages. Because this set need not be a singleton, under which distribution should the designer evaluate the mechanism and possible deviations? One natural response is to adopt maximin or regret-based criteria, in the spirit of robust design, over this set. Such criteria can readily be incorporated into our framework, but they impose additional feature: the putative mechanism and its deviations may be evaluated under different adverse distributions rather than under the same belief. We instead adopt an equilibrium approach. We call a mechanism stable if it is optimal for the designer under some type distribution consistent with the empirical distribution of messages. In other words, a stable mechanism is sustained by at least one belief compatible with the data it generates.

A fourth issue is that the equilibrium concept can remain permissive under the consistency requirement between beliefs and mechanisms specified above:  mechanisms that reveal little information may leave many type distributions observationally indistinguishable, making it easier to find a belief that sustains them. The equilibrium approach naturally permits refinements, following the classical tradition in game theory. We say that a mechanism is robustly stable if it remains stable under every sufficiently small exogenous revelation of information about the true distribution. This information leakage, or “a grain of truth,” is modeled by revealing the true distribution on a set that has small probability under the true distribution. The term “robustness” is therefore used in its classical sense—stability under perturbations—rather than in the more recent sense associated with robust mechanism design. The refinement therefore asks whether a mechanism remains stable under sufficiently small leakage of truthful information.

Regarding equilibrium multiplicity, we again need to clarify a possible source of confusion. Although restricting attention to direct mechanisms is not without loss of generality once information and inference matter, as discussed above, might it nevertheless be without loss of optimality? Notice that the mechanism designer is herself a player in this problem. The question therefore amounts to asking whether a player can select the equilibrium most favorable to herself. The question is therefore not well posed and cannot be settled by the usual argument that more information is valuable. What is true is that the optimal direct mechanism is indeed stable with respect to the true distribution, in the sense defined above, and cannot be refined away. But it is only one among potentially many stable mechanisms. The fundamental question is which other mechanisms are stable. In our setting, information is not what the designer can choose independently of the mechanism. Mechanism choice and information are jointly determined, and a stable mechanism is a fixed point of this interaction. As in a game with multiple equilibria, the equilibrium concept characterizes the possible stable mechanisms but does not by itself select the one most favorable to a particular player.

From this perspective, in the classic buyer–seller problem, a natural and practically relevant question is the following: is a posted-price mechanism—one that reveals only a single value of the cumulative distribution function of the underlying types—stable? If so, which prices can be stable? This is a well-posed question of practical relevance. More broadly, among the many possible mechanisms, which can arise as stable mechanisms, and how can they be characterized?

With these conceptual issues clarified, the scope and limits of our framework should now be clearer. We develop this theory formally and prove a new revelation principle, which we call the fictitious revelation principle. We show that any incentive compatible mechanism is represented by a fictitious direct mechanism in which agents truthfully report types, but these reports are  filtered to add noise. Importantly, the relevant notion of informational equivalence is weaker than Blackwell equivalence. This construction allows us to represent all mechanisms---direct or indirect---within a unified framework that preserves both outcome and informational content. 
With the fictitious revelation principle in place,
we formulate stability and its refinement in terms of
fictitious direct mechanisms.
We develop monotonicity results and a verification
criterion based on least informative filters to
characterize robustly stable mechanisms.

We then apply the framework to the classic monopoly
problem. We show that a posted-price mechanism is
robustly stable precisely when its price is a local
maximizer of revenue. This identifies which suboptimal
prices can persist.
Under the usual regularity condition (monotone marginal
revenue), or more generally when revenue is single-peaked,
robust stability coincides with optimality among
incentive compatible and individually rational mechanisms,
as if the true prior were known to the designer.
Under these conditions, if the optimal price is unique,
the Myerson mechanism is the unique robustly stable
posted-price mechanism.
In this sense, informational feedback eliminates
suboptimal stable mechanisms.

\medskip

\noindent\textbf{Related Literature}

\medskip

\noindent Our benchmark concept of stability is closely related to the game-theoretic solution concepts of self-confirming equilibrium \citep{fudenberg1993self} and conjectural equilibrium \citep{battigalli1988conjectural,rubinstein1994rationalizable}; see also \citet{cho2008self} and \citet{battigalli2022framework} for applications of these ideas to economic policy analysis. Although a mechanism can be viewed as an action or policy chosen by the designer, the flexibility to design message spaces and outcome rules, subject to agents' incentive constraints, is central to mechanism design. Much of this structure would be lost if the problem were treated solely as decision analysis. The self-confirming equilibrium concepts require players’ behavior to be optimal given conjectures about the uncertain environment that are empirically consistent with the observations generated by play. Our setting differs in that the designer is assumed to know the agents’ strategy profile, which constitutes an ex post or dominant-strategy equilibrium of the induced game, but is uncertain about the distribution of their types; likewise, agents have no uncertainty about the designer's mechanism choices.
The broader idea of using observations to impose consistency on beliefs and thereby define a solution concept also appears in  \citet{liu2023cooperative}.

Our focus is a reduced-form stability criterion for general mechanisms. Our motivation, however, is closely related to the growing literature of data-driven theory that explicitly models learning and policy experimentation.
\citet{ortner2018making} develop a prior-free approach to local policy evaluation in a principal--agent--monitor environment. They use equilibrium reporting data to identify local policy improvements and propose small experiments when existing policy variation is insufficient. They also discuss iterative policy improvement, which leads to a local optimum if it converges, and to the global optimum under additional conditions. Related approaches use equilibrium data to draw policy-relevant inferences \citep{chassang2019collusion,chassang2019crime}.
\citet{cho2025learning} study when a simple learning algorithm can identify the optimal monopoly price even when demand is nonlinear. \citet{malladi2026guarantees} studies price experimentation under unknown demand and shows that sequential worst-case maximization leads to rising prices.

Although our paper is not about robust mechanism design per se, it speaks to the same broad setting in which the designer lacks precise information about the design environment. The fact that maximin criteria may select weakly dominated mechanisms or mechanisms with seemingly arbitrary inefficiencies has motivated both refinements and alternative approaches. \citet{kambhampati2024proper} develops the notion of proper robustness, which iteratively applies the worst-case criterion to refine the set of maximin-optimal mechanisms. By contrast, \citet{BorgersLiWang2025} compare mechanisms pointwise across type profiles, ruling out a mechanism if another yields weakly higher revenue at every type profile and strictly higher revenue at some profile. Neither approach endogenizes the designer’s information through the operation of the mechanism, which is central to our framework.

\section{Basics}\label{sec:mechanisms}

Consider a finite set of agents \( N = \{1,2,\dots,n\} \). Each agent 
\( i \in N \) has a private type \( \theta^i \in \Theta^i \), and the type space is 
\( \Theta = \prod_{i \in N} \Theta^i \). A type profile is denoted by 
\( \theta = (\theta^1,\theta^2,\dots,\theta^n) \in \Theta \). 
Let \( O \) be the set of outcomes.

Each agent \( i \in N \) has a payoff function 
\( u^i : O \times \Theta \to \mathbb{R} \). 
The mechanism designer is modeled as player \( 0 \), with payoff function 
\( u^0 : O \times \Theta \to \mathbb{R} \). 
These payoff functions are extended in the standard expected-utility fashion 
to probability measures on \( O \times \Theta \); with a slight abuse of notation, 
we continue to denote the extended functions by \( u^i \) and \( u^0 \).

A \textbf{mechanism} \( \mathcal{M} = (M,\omega) \) consists of a message space 
\( M = \prod_{i \in N} M^i \), where \( M^i \) is the set of messages available 
to agent \( i \), and an outcome function 
\( \omega : M \to \Delta(O) \) 
that assigns to each message profile 
\( m = (m^1,m^2,\dots,m^n) \) 
a probability measure over outcomes.

\begin{remark}[Topology and Measurability]
We assume that all type, message, and outcome spaces are uncountable standard Borel spaces.\qed
\end{remark}

\subsection{Mathematical Notation}

For any stochastic mapping \( g : Y \to \Delta(Z) \) and any probability 
measure \( \mu \in \Delta(Y) \), we define \( g_{\mu} \in \Delta(Z) \) as 
the \textbf{average} of \( g \) with respect to \( \mu \), given by
\begin{equation}\label{eq:average}
    g_{\mu}(E) := \int_Y g(y)(E)\, \mu(dy)
\end{equation}
for any measurable set \( E \subseteq Z \).

For any two stochastic mappings \( h : X \to \Delta(Y) \) and 
\( g : Y \to \Delta(Z) \), we define their (kernel) \textbf{compound} 
\( g \circ h : X \to \Delta(Z) \) by
\begin{equation}\label{eq:compound1}
    (g \circ h)(x) := g_{h(x)}
\end{equation}
for all \( x \in X \). Equivalently, using \eqref{eq:average},
\begin{equation}\label{eq:compound2}
    (g \circ h)(x)(E) 
    = \int_Y g(y)(E)\, h(x)(dy)
\end{equation}
for all \( x \in X \) and measurable set \( E \subseteq Z \).

For a collection of functions \( h^i : X^i \to \Delta(Y^i) \), 
\( i = 1,\dots,n \), define their \textbf{joint} function
\begin{equation}\label{eq:joint-function}
h(x) := ( h^1(x^1), \dots, h^n(x^n) ),    
\end{equation}
for all \( x = (x^1,\dots,x^n) \in X = \prod_{i \in N} X^i \).
Thus, \( h \) may be viewed as a function from 
\( X\) to \(\Delta(Y)=\Delta(\prod_{i \in N}Y^i) \supseteq \prod_{i \in N} \Delta(Y^i)\).

\subsection{Mechanism-Induced Game}

The designer is a player in the game, who chooses a mechanism \( \mathcal{M} \), which is publicly observed. The chosen mechanism \( \mathcal{M} \) and a profile of agents' payoff functions 
\( u = (u^1,\dots,u^n) \) induce a \textit{subgame} \( (\mathcal{M},u) \) among the 
\( n \) agents. A \textbf{strategy} of agent \( i \) in this game is a 
function \( \sigma^i : \Theta^i \to \Delta(M^i) \).

As in the standard interpretation of mechanism design, each agent 
\( i \in N \) sends a message \( m^i \in M^i \), possibly at random, 
according to \( \sigma^i \). The realized message profile 
\( m = (m^1,\dots,m^n) \in M \) is mapped by the outcome function $\omega$ to a 
probability measure over outcomes, \( \omega(m) \in \Delta(O) \).

In accordance with the joint function notation in \eqref{eq:joint-function}, a strategy profile  \( \sigma = (\sigma^1,\dots,\sigma^n) \) can be regarded as  
 a stochastic mapping \( \sigma : \Theta \to \Delta(M) \).
The compound of \( \sigma \) and the outcome function 
\( \omega : M \to \Delta(O) \) is therefore a stochastic mapping 
\( \omega \circ \sigma : \Theta \to \Delta(O) \), 
which assigns to each type profile a probability measure over outcomes.

With a slight abuse of notation, we write 
\( u^i(\omega \circ \sigma, \theta) \) in place of 
\( u^i((\omega \circ \sigma)(\theta), \theta) \). A strategy profile \( \sigma = (\sigma^1,\dots,\sigma^n) \) is an 
\textbf{ex post equilibrium} of the induced game if
\[
u^i(\omega \circ \sigma, \theta)
\;\ge\;
u^i(\omega \circ (\tilde{\sigma}^i,\sigma^{-i}), \theta)
\]
for all \( i \in N \),  \( \theta \in \Theta \), 
and  \( \tilde{\sigma}^i : \Theta^i \to \Delta(M^i) \).

We say that \( \sigma^i \) is a \textbf{dominant strategy}, or an 
``always best response,'' for agent \( i \) in the induced game if
\[
u^i(\omega\circ (\sigma^i,\tilde{\sigma}^{-i}), \theta)\geq u^i(\omega \circ\tilde{\sigma}, \theta)
\]
for all \( \theta \in \Theta \) and all strategy profiles 
\( \tilde{\sigma} = (\tilde{\sigma}^1,\dots,\tilde{\sigma}^n) \).

A strategy profile \( \sigma = (\sigma^1,\dots,\sigma^n) \) is a 
\textbf{dominant strategy equilibrium} of the induced game if 
\( \sigma^i \) is a dominant strategy for every \( i \in N \).

\begin{definition}
    We call a mechanism-strategy pair \( (\mathcal{M}, \sigma)  \) an \textbf{augmented mechanism}. An augmented mechanism is (ex post or dominant strategy) \textbf{incentive compatible} if \( \sigma \) constitutes an (ex post or dominant strategy) equilibrium of the game induced by the mechanism $\mathcal{M}$.

\end{definition}

\subsection{Mechanism-Generated Data and Inference}\label{sec:data}

Given any probability measure $\pi \in \Delta(\Theta)$ and any strategy profile $\sigma$, 
the probability measure induced by $\pi$ and $\sigma$ over the message space $M$ 
is denoted by $\sigma_{\pi}$, following the notation in \eqref{eq:average}.

Assume that nature chooses agents' types according to a prior 
$\pi_0 \in \Delta(\Theta)$, which we refer to as the \textbf{true prior}. 
We assume that no one---especially the designer---observes $\pi_0$. 
Instead, they observe the \textbf{empirical distribution} of message profiles 
induced by the strategy profile $\sigma$ and the prior $\pi_0$, namely 
$\sigma_{\pi_0} \in \Delta(M)$ according to the notational convention in \eqref{eq:average}. 

The basic premise is that the designer knows the agents' strategy profile 
$\sigma$ and observes the empirical distribution of messages 
$\sigma_{\pi_0}$, from which she aims to infer $\pi_0$. 
However, given $\sigma$, there may exist multiple 
$\pi \in \Delta(\Theta)$ such that $\sigma_\pi = \sigma_{\pi_0}$, 
all of which are observationally indistinguishable from the true prior $\pi_0$.  %
We write $\Pi^{\mathcal{M}}(\sigma, \sigma_{\pi_0}) \subseteq \Delta(\Theta)$ 
for the set of such probability measures, which we term \textbf{feasible priors}.
Formally,
\begin{equation}\label{eq:feasible-priors}
   \Pi^{\mathcal{M}}(\sigma,\sigma_{\pi_0}) :=  \left\{\pi \in \Delta(\Theta): \sigma_{\pi}=\sigma_{\pi_0} \right\} .
\end{equation}
It is a convex subset of $\Delta(\Theta).$ Equivalently, $\Pi^{\mathcal{M}}(\sigma,\sigma_{\pi_0})$ contains all 
$\pi \in \Delta(\Theta)$ such that $\sigma_{\pi_0}$ is a $\sigma$-garbling of $\pi$. Let us first consider two examples.

\begin{example}\label{eg:direct}
If \(\mathcal{M}\) is a direct mechanism and $\sigma$ is a truth-telling strategy profile, then $\sigma_{\pi}=\pi$ for any $\pi \in \Delta(\Theta).$ Hence, $\Pi^{\mathcal{M}}(\sigma,\sigma_{\pi_0})=\{\pi_0\}$ for any true prior $\pi_0 \in \Delta(\Theta)$. The augmented mechanism fully reveals the true prior distribution and produces the maximum amount of information.\qed 
\end{example}

\begin{example}\label{eg:pooling}
    For any mechanism $\mathcal{M}$, suppose agents play a pooling strategy profile $\sigma \equiv 1_{m^*}$, where $m^* \in M$ and $1_{m^*}$ denotes the Dirac measure on $m^*$. Then $\sigma_{\pi}=1_{m^*}\in \Delta(M)$ for any $\pi \in \Delta(\Theta).$ Hence, $\Pi^{\mathcal{M}}(\sigma, \sigma_{\pi_0})=\Delta(\Theta)$ for any true prior $\pi_0 \in \Delta(\Theta)$. That is, the augmented mechanism produces the minimal amount of information.\qed
\end{example}

The set of feasible priors observationally indistinguishable from the true prior, \eqref{eq:feasible-priors} is  a central piece in our endeavor of theory building. Future extensions and refinements of the theory will naturally build on this benchmark. One could impose additional restrictions on candidate priors, such as requiring agents’ types to be i.i.d. across agents. While there are many possibilities to explore, we highlight two particularly important considerations.

\begin{remark}
One may ask why we take as our benchmark a model in which the designer observes the distribution of messages. The motivation is simple: if a mechanism passes the stability test and persists over time, the designer will eventually observe this distribution. As discussed in the introduction, one alternative interpretation is that the mechanism is used repeatedly with successive generations of agents whose types are drawn from a common distribution, and our reduced-form framework abstracts from statistical issues such as sampling noise. This modeling choice is, of course, an idealization. Alternatively, one could assume that the designer observes only a summary statistic of the message distribution, in which case the condition $\sigma_{\pi}=\sigma_{\pi_0}$ in \eqref{eq:feasible-priors} would be replaced by equality of the corresponding statistics. This would weakly enlarge the set of type distributions consistent with the observation but introduce no fundamentally new conceptual ideas, only additional notation.\qed
\end{remark}

\begin{remark}
    We also assume that the agents' strategy profile $\sigma$ is known to the designer. This assumption is similar to the standard equilibrium postulate. When the subgame induced by the mechanism admits an ex post or dominant-strategy equilibrium, as we assume here, the assumption is relatively mild. If one wishes to allow uncertainty about $\sigma$, the condition $\sigma_{\pi}=\sigma_{\pi_0}$ in \eqref{eq:feasible-priors} could instead be replaced by $ {\sigma}_{\pi}=\sigma_{\pi_0}$ for some strategy profile $ {\sigma}$. This would raise a number of additional problems. More importantly, it would enlarge the set of feasible priors without shedding much additional light on the central question of this paper.\qed
\end{remark}

As demonstrated above, we can compare two augmented mechanisms through the set of feasible priors they induce. This naturally leads to the following definition.

\begin{definition}
\label{def:xinformative}
Consider two stochastic mappings $g: X \to \Delta(Y)$ and 
$h: X \to \Delta(Z)$. 
We say that $g$ is \textbf{kernel more informative} than $h$, 
denoted $g \succsim h$, if for every $\mu_0 \in \Delta(X)$,
\begin{equation}\label{eq:informationorder}
\{\mu \in \Delta(X) : g_{\mu} = g_{\mu_0}\}
\subseteq
\{\mu \in \Delta(X) : h_{\mu} = h_{\mu_0}\}.
\end{equation}
We say that $g$ and $h$ are \textbf{kernel equivalent}, 
denoted $g \sim h$, if $g \succsim h$ and $h \succsim g$.
\end{definition}

   The term ``kernel'' reflects the fact that $g$ and $h$ can be viewed as 
linear operators on measures, and the set of feasible priors is closely 
related to the kernel (null space) of the induced linear operator. 
We drop the qualifier ``kernel'' when the intended informativeness order 
is clear from the context. The relationship between this order and the Blackwell informativeness order is discussed in Online Appendix \ref{sec:Kernel-vs-Blackwell}: 
if $g$ is Blackwell more informative than $h$, then $g \succsim h$. 
Thus, the notion of kernel informativeness in \Cref{def:xinformative} is weaker than Blackwell informativeness. 
 
This new information order was independently developed by
\citet{rosenthal2026prior} in the finite case. In his setting, the decision
maker observes the population distribution of signals generated by a
Blackwell experiment under an unknown state distribution. Experiments
are ranked by their worst-case value across decision problems. He shows
that this decision-theoretic ranking is characterized by this order
and it is weaker than the Blackwell order. Building on this idea,
\citet{rosenthal2025prior} studies information design, characterizing
implementable actions and the most informative information structures
that implement them.

\subsection{Basic Criterion of Stability}
We begin with a basic stability criterion: a mechanism must be justifiable under some belief consistent with the empirical observations it generates, a natural requirement for its continued use. This criterion serves as a benchmark, as we have yet to explore the rich structure offered by mechanisms. We later develop tools for analyzing stability and introduce a refinement that strengthens the test.

Given a mechanism $\mathcal{M}=(M,\omega)$, a strategy profile $\sigma$, 
and a probability measure $\pi\in \Delta(\Theta)$ over type profiles, we write the designer's expected payoff as
\begin{equation}\label{eqn:V0}
V^0(\mathcal{M},\sigma,\pi)
=
\int_{\Theta} u^0((\omega \circ \sigma)(\theta), \theta)\,\pi(d\theta).
\end{equation}

\begin{definition}\label{def:stable-vanila}
An incentive compatible augmented mechanism $(\mathcal{M},\sigma)$ is 
$\pi_0$-\textbf{stable} if there exists $\pi \in 
\Pi^{\mathcal{M}}(\sigma,\sigma_{\pi_0})$ such that
\begin{equation}\label{ineq:self-confirming}
V^0(\mathcal{M},\sigma,\pi) 
\ge 
V^0(\mathcal{M}',\sigma',\pi)
\end{equation}
for all incentive compatible augmented mechanisms $(\mathcal{M}',\sigma')$.
\end{definition}

A stable mechanism does not depend on the true prior $\pi_0$ directly, but instead depends on the empirical distribution $\sigma_{\pi_0}$; for notational simplicity, we explicitly introduce $\pi_0$ into the definition. Two remarks are immediate.

\begin{remark}
We have not explicitly incorporated individual rationality, limited liability, or other constraints into our formulation, although these can readily be added as needed. Consistent with our use of ex post or dominant-strategy incentive compatibility, the appropriate notion of individual rationality is ex post individual rationality.\qed
\end{remark}

\begin{remark}\label{rem:maxmin}
    The idea of information-generating mechanisms opens up many possibilities for alternative solution concepts, such as requiring $(\mathcal{M},\sigma)$ to solve 
       \begin{equation*}%
     \sup_{(\mathcal{M}',\sigma')}\inf_{\pi \in \Pi^{\mathcal{M}}(\sigma,\sigma_{\pi_0})} V^0(\mathcal{M}',\sigma',\pi)
   \end{equation*} or, equivalently, for any incentive compatible augmented mechanism $(\mathcal{M}',\sigma')$,
      \begin{equation}\label{eq:left-right-maximin} 
   \inf_{\pi \in \Pi^{\mathcal{M}}(\sigma,\sigma_{\pi_0})} V^0(\mathcal{M},\sigma,\pi) \geq \inf_{\pi \in \Pi^{\mathcal{M}}(\sigma,\sigma_{\pi_0})} V^0(\mathcal{M}',\sigma',\pi).
   \end{equation}
    This alternative evaluates the putative mechanism and each deviation by their respective worst-case payoffs over the same set of empirically consistent beliefs. The worst-case belief may therefore differ across the putative mechanism (the LHS of \eqref{eq:left-right-maximin}) and its deviations (the RHS of \eqref{eq:left-right-maximin}), which is quite different from the standard equilibrium analysis.\qed %
\end{remark}

\section{Fictitious Revelation Principle}\label{sec:FRP}
\subsection{Direct Mechanisms}
A direct mechanism $\mathcal{D}=(\Theta,\delta)$ is a special mechanism in which 
the message space $M$ coincides with the type space $\Theta$, and the outcome 
function $\delta:\Theta \to \Delta(O)$ maps reported types to outcomes. For each agent $i$, a truth-telling strategy in a direct mechanism is the 
identity mapping  $\tau^i: \theta^i \mapsto \theta^i$. The profile of truth-telling strategies is the joint function 
$\tau=(\tau^1,\ldots,\tau^n)$.

Whenever the context is clear, we say that a direct mechanism 
$\mathcal{D}=(\Theta,\delta)$ is (ex post or dominant-strategy) 
\textbf{incentive compatible} if the profile of truth-telling strategies 
constitutes an (ex post or dominant-strategy) equilibrium of the induced game, 
without explicitly referring to the augmentation of the direct mechanism 
by the truth-telling strategies.

Due to the classical revelation principle, incentive-compatible direct mechanisms are often used by analysts as a technical device to encompass, without loss of generality, all incentive-compatible augmented mechanisms. 
However, in our setting, as demonstrated by \Cref{eg:direct}, the designer can identify the true prior $\pi_0$ from the empirical distribution of messages generated by a direct mechanism with the truth-telling strategy. By contrast, as shown in \Cref{eg:pooling}, the designer cannot infer anything about the true prior from the empirical distribution generated by a pooling strategy.
Clearly, there exist other augmented mechanisms that lie between these two extremes in terms of the feasible priors they induce, as defined in \eqref{eq:feasible-priors}. As a result, a naive application of the classical revelation principle fails. To establish a revelation principle applicable in our setting with unknown priors and observable empirical message distributions, we must define a modified class of direct mechanisms that preserves informational content.

\subsection{Fictitious Direct Mechanisms}\label{sec:fictitious}

In this section, we introduce what we call fictitious direct mechanisms. It is useful to keep in mind that these are not mechanisms themselves, but rather technical devices for theoretical analysis, to capture the information content from an arbitrary augmented mechanism, as defined in \eqref{eq:feasible-priors}.

We equip a direct mechanism \(\mathcal{D}=(\Theta, \delta)\) with a profile of stochastic mappings $\varphi=(\varphi^1,\ldots,\varphi^n)$, where each
\[
\varphi^i:\Theta^i \to \Delta(\Theta^i).
\] By the convention in \eqref{eq:joint-function}, the joint function \(
\varphi:\Theta \to \Delta(\Theta)
\) is a stochastic mapping.  
The interpretation is that if \( \theta \) is the reported type profile in a direct mechanism, a \textbf{fictitious signal} is drawn from the probability measure \( \varphi(\theta) \in \Delta(\Theta) \). So $\varphi$ is a garbling of the truth-telling strategy $\tau$. %

Suppose the direct mechanism $\mathcal{D}$ is run under the unknown true prior $\pi_0$, with agents following the truth-telling strategy $\tau$. Rather than taking the reported types themselves as the observations, consider fictitious signals generated by applying the stochastic mapping $\varphi$ to these reports. The resulting \textbf{fictitious empirical distribution} of signals is $\varphi_{\pi_0}\in\Delta(\Theta)$, defined for any measurable set $E\subseteq\Theta$ by
\begin{equation}
    \varphi_{\pi_0}(E)
    := \int_{\Theta} \varphi(\theta)(E)\,\pi_0(d\theta).
\end{equation}

The set of feasible priors that generate the same fictitious empirical distribution as the true prior $\pi_0$ is given by
\begin{equation}\label{eq:fictitious-priors}
    \Pi^{\mathcal{D}}(\varphi,\varphi_{\pi_0})
    := \{\pi \in \Delta(\Theta) : \varphi_{\pi} = \varphi_{\pi_0} \}.
\end{equation}
In pursuing a new revelation principle, the central question is whether \eqref{eq:fictitious-priors} can exactly recover the feasible priors in \eqref{eq:feasible-priors} induced by an augmented mechanism $(\mathcal{M},\sigma)$ for any unknown true prior; that is, whether the following holds:
\begin{equation}\label{eq:info-equivalent}
\Pi^{\mathcal{D}}(\varphi,\varphi_{\pi_0})
=\Pi^{\mathcal{M}}(\sigma,\sigma_{\pi_0}) \textup{ for all } \pi_0 \in \Delta(\Theta).
\end{equation}
Using the terminology of \Cref{def:xinformative}, \eqref{eq:info-equivalent} is the same as saying $\varphi$ and $\sigma$ are kernel equivalent, i.e., $\varphi \sim \sigma$. 

Not all stochastic mapping $\varphi:\Theta \rightarrow \Delta(\Theta)$ will do the job. We need a Blackwell informativeness restriction. %

\begin{definition}
\label{def:fictitious}
Consider a direct mechanism $\mathcal D=(\Theta,\delta)$
and a profile of stochastic mappings
$\varphi=(\varphi^1,\ldots,\varphi^n)$, where
\[
\varphi^i:\Theta^i\to\Delta(\Theta^i)
\quad\text{for each }i\in N.
\]

\begin{enumerate}
    \item[\textup{(i)}]
    We call $\varphi$ a \textbf{filter} of $\mathcal D$
    if it is weakly Blackwell more informative than
    $\delta$; that is, there exists a stochastic mapping
    $g:\Theta\to\Delta(O)$ such that
    \[
    \delta=g\circ\varphi.
    \]

    \item[\textup{(ii)}]
    If $\varphi$ is a filter of $\mathcal D$,
    we call $(\mathcal D,\varphi)$ a
    \textbf{fictitious direct mechanism}.

    \item[\textup{(iii)}]
    We call a fictitious direct mechanism
    $(\mathcal D,\varphi)$
    (ex post or dominant-strategy)
    \textbf{incentive compatible}
    if the truth-telling strategy profile constitutes an (ex post or dominant-strategy) equilibrium of the game induced by $\mathcal D$.
\end{enumerate}
\end{definition}

    Fictitious direct mechanisms are direct mechanisms coupled with filters, so they are \textit{not} mechanisms themselves, but purely technical devices aimed to capture all augmented mechanisms $(\mathcal{M},\sigma)$ in terms of both allocations and potentially very complex information structures. The description above already suggests a different interpretation. A third party runs a standard direct mechanism but reveals filtered information to the designer or, alternatively, the designer runs the direct mechanism and commits to observing the reported types through $\varphi.$ Note that the outcome rule is applied directly to reported types,
while the filter determines only the information
available to the designer. This \textit{separation} makes
fictitious direct mechanisms particularly easy to work with.

    Since messages in $(\mathcal{M},\sigma)$ reveal more than the outcomes (an outcome function is defined on the message space), the filter $\varphi$ in the fictitious direct mechanism that captures the information of $\mathcal{M}$ should be at least as informative as $\delta$ in the direct mechanism $\mathcal{D}$. A natural benchmark notion of informativeness is the Blackwell order. This motivates the requirement for a filter in \Cref{def:fictitious}.

\subsection{Fictitious Revelation Principle}
\label{sec:fictitiousrp}

 An augmented mechanism $(\mathcal{M}, \sigma)$ and a fictitious direct mechanism are equivalent if they implement the same allocations and generate the same information.

 \begin{definition}
An augmented mechanism $(\mathcal M,\sigma)$, where $\mathcal M=(M,\omega)$, and
a fictitious direct mechanism $(\mathcal D,\varphi)$,  where $\mathcal D=(\Theta,\delta)$,
are \textbf{equivalent} if the following two conditions are satisfied:
\begin{enumerate}
    \item[\textup{(i)}]
    Outcome equivalence:
    $\omega\circ\sigma=\delta$;

    \item[\textup{(ii)}]
    Information equivalence:
    $\sigma\sim\varphi$.
\end{enumerate}
In this case, we say that $(\mathcal M,\sigma)$
\textbf{implements} $(\mathcal D,\varphi)$, and that
$(\mathcal D,\varphi)$ \textbf{represents}
$(\mathcal M,\sigma)$.
\end{definition}

The second condition above implies that no matter what the \textit{true} prior distribution is, the filter $\varphi$ generates a \textit{fictitious} observation that induces the same set of feasible priors as the empirical distribution generated by $\sigma$ does. 
\Cref{fig:revelation-principles} compares the standard
and fictitious revelation principles.

\begin{figure}[htbp]
\centering
\captionsetup{font=small,skip=6pt}
\captionsetup[subfigure]{skip=4pt}

\begin{subfigure}[b]{0.48\linewidth}
\centering
\begin{tikzpicture}[
    every node/.style={font=\normalsize},
    arr/.style={->, thick, >=stealth}
]
\node (theta0) at (0,0) {$\Theta$};
\node (theta1) at (3,1.2) {$\Theta$};
\node (M)      at (3,-1.2) {$M$};
\node (out)    at (6,0) {$O$};

\draw[arr]
    (theta0) -- node[above left] {$\tau$} (theta1);
\draw[arr,dash pattern=on 3pt off 2pt]
    (theta0) -- node[below left] {$\sigma$} (M);
\draw[arr,dash pattern=on 3pt off 2pt]
    (M) -- node[below right] {$\omega$} (out);
\draw[arr,dash pattern=on 3pt off 2pt]
    (theta1) -- node[above right]
    {$\delta=\omega\circ\sigma$} (out);
\end{tikzpicture}
\caption{}
\end{subfigure}
\hfill
\begin{subfigure}[b]{0.48\linewidth}
\centering
\begin{tikzpicture}[
    every node/.style={font=\normalsize},
    arr/.style={->, thick, >=stealth}
]
\node (theta0) at (0,0) {$\Theta$};
\node (theta1) at (3,1.2) {$\Theta$};
\node (theta2) at (3,2.6) {$\Theta$};
\node (M)      at (3,-1.2) {$M$};
\node (out)    at (6,0) {$O$};

\draw[arr]
    (theta0) -- node[above left] {$\tau$} (theta1);
\draw[arr,dash pattern=on 3pt off 2pt]
    (theta1) -- node[right]
    {$\varphi\sim\sigma$} (theta2);
\draw[arr,dash pattern=on 3pt off 2pt]
    (theta0) -- node[below left] {$\sigma$} (M);
\draw[arr,dash pattern=on 3pt off 2pt]
    (M) -- node[below right] {$\omega$} (out);
\draw[arr,dash pattern=on 3pt off 2pt]
    (theta1) -- node[above right]
    {$\delta=\omega\circ\sigma$} (out);
\end{tikzpicture}
\caption{}
\end{subfigure}

\caption{\footnotesize{
(a) The standard representation preserves outcomes,
$\delta=\omega\circ\sigma$.
(b) The fictitious representation also preserves
information, $\varphi\sim\sigma$.
Dashed arrows denote stochastic kernels, and
$\circ$ denotes kernel composition.}}
\label{fig:revelation-principles}
\end{figure}

In many applications, we are interested in deterministic mechanisms and strategies. Mathematically, we regard a Borel measurable function $g:X \rightarrow Y$ as a special stochastic mapping $x \mapsto \bm{1}_{g(x)}$.

\begin{definition}[Deterministic Mechanisms]\label{def:deterministic}
\leavevmode
\begin{enumerate}
    \item[\textup{(i)}]
    An augmented mechanism $(\mathcal M,\sigma)$,
    where $\mathcal M=(M,\omega)$, is
    \textbf{strategically deterministic} if $\sigma$
    is a pure strategy profile, and \textbf{deterministic}
    if, in addition, $\omega$ is deterministic.

    \item[\textup{(ii)}]
    A fictitious direct mechanism $(\mathcal D,\varphi)$,
    where $\mathcal D=(\Theta,\delta)$, is
    \textbf{informationally deterministic} if $\varphi$
    is deterministic, and \textbf{deterministic}
    if, in addition, $\delta$ is deterministic.
\end{enumerate}
\end{definition}

The following fictitious revelation principle relates
strategically deterministic augmented mechanisms to
informationally deterministic fictitious direct mechanisms.
It applies, in particular, to fully deterministic mechanisms.

\begin{theorem}[Deterministic Fictitious Revelation Principle]
\label{thm:deterministic-RP}
\leavevmode
\begin{enumerate}
    \item[\textup{(i)}]
    Every (informationally) deterministic fictitious direct mechanism
    $(\mathcal D,\varphi)$ is implemented by some (strategically) deterministic
    augmented mechanism $(\mathcal M,\sigma)$.
    Conversely, every (strategically) deterministic augmented mechanism
    $(\mathcal M,\sigma)$ is represented by some (informationally) deterministic
    fictitious direct mechanism $(\mathcal D,\varphi)$.

    \item[\textup{(ii)}]
    Both assertions in \textup{(i)} remain valid when
    incentive compatibility is imposed on both mechanisms,
    in either the ex post or the dominant-strategy sense.
\end{enumerate}
\end{theorem}

The fictitious revelation principle for stochastic mechanisms, however, holds only in the direction of representation if incentive compatibility is imposed.

\begin{theorem}[Stochastic Fictitious Revelation Principle]
\label{thm:stochastic-RP}
\leavevmode
\begin{enumerate}
    \item[\textup{(i)}]
    Every fictitious direct mechanism $(\mathcal D,\varphi)$
    is implemented by an augmented mechanism
    $(\mathcal M,\sigma)$.
    Conversely, every augmented mechanism $(\mathcal M,\sigma)$
    is represented by a fictitious direct mechanism
    $(\mathcal D,\varphi)$.

    \item[\textup{(ii)}]
   Every incentive compatible augmented mechanism is represented
by an incentive compatible fictitious direct mechanism,
but the converse need not hold.
\end{enumerate}
\end{theorem}

For incentive compatible mechanisms, the stochastic fictitious revelation principle is one-directional. There are incentive compatible stochastic fictitious mechanisms that are not implementable in incentive compatible augmented mechanisms. Online Appendix \ref{sec:incentive-implementability}  provides an example.
Stochastic fictitious direct mechanisms however serve as a convenient technical device: we first search over incentive compatible fictitious direct mechanisms and then check implementability by an incentive compatible augmented mechanism. 

Motivated by \Cref{thm:stochastic-RP}\textup{(ii)},
we introduce the following restriction on fictitious
direct mechanisms so that the correspondence with
incentive compatible augmented mechanisms holds
in both directions.

\begin{definition}
A fictitious direct mechanism $(\mathcal D,\varphi)$
is \textbf{incentive implementable} if it admits an implementation
by an incentive compatible augmented mechanism.
\end{definition}

\begin{remark}
Incentive implementability often needs to be verified
only for a least informative filter, as shown in \Cref{thm:least-informative-filter}
and illustrated in \Cref{sec:application}. It implies incentive compatibility, which refers only to
truth-telling in the underlying direct mechanism
$\mathcal D$, as defined in
\Cref{def:fictitious}\textup{(iii)}. Indeed, any type misreport
in $\mathcal D$ can be replicated in its incentive
compatible implementation $(\mathcal M,\sigma)$
by using the corresponding mixed message by the outcome equivalence,
$\delta=\omega\circ\sigma$. Incentive compatibility
of the implementation implies that of $\mathcal D$,
in either the ex post or the dominant-strategy sense. It follows from \Cref{thm:deterministic-RP} that all incentive compatibility together with informational determinism imply incentive implementability. \qed
\end{remark}

We now explain why incentive compatibility and
incentive implementability coincide when the filter
is deterministic, but when stochastic filters are allowed, incentive compatibility
is strictly weaker than incentive implementability.
Thus the incentive compatible revelation principle
holds in both directions in \Cref{thm:deterministic-RP},
but only in one direction in \Cref{thm:stochastic-RP}.

Since $\varphi$ is Blackwell more informative than $\delta$,
there exists a Borel stochastic mapping
$g:\Theta\to\Delta(O)$ such that $\delta=g\circ\varphi$.
Setting $M=\Theta$, $\omega=g$, and $\sigma=\varphi$
therefore implements $(\mathcal D,\varphi)$, although this
construction does not automatically preserve incentive
compatibility.

When $\varphi$ is deterministic, we have
$\delta(\theta)=g(\varphi(\theta))$.
Every deviation to a used message
$m^i\in\varphi^i(\Theta^i)$ can be replicated by a
direct misreport $\hat\theta^i$ satisfying
$m^i=\varphi^i(\hat\theta^i)$.
Direct incentive compatibility therefore rules out
profitable deviations within the used message ranges.
Unused messages require an additional modification of
the outcome function (indirectly restricting the message spaces), subject to a Borel measurability requirement.
This argument permits stochastic outcome rules.

With stochastic filters, however, a direct misreport
selects a prescribed lottery over messages,
$\varphi^i(\hat\theta^i)$, whereas an agent in the
augmented mechanism can deviate to an individual message
$m^i$.
Such a deviation need not be reproducible by any direct
misreport, even when that message is used under the
prescribed strategies.

\section{Stable Mechanisms}\label{sec:solution-concept}
\subsection{Fictitious Representation}\label{sec:selfconfirming}

In this section, we represent stable augmented mechanisms by fictitious direct mechanisms, setting the stage for refinement and the development of further analytical tools for studying stable mechanisms.

Given a direct mechanism $\mathcal{D}=(\Theta,\delta)$, 
the agents' truth-telling strategy profile, and a probability measure 
$\pi\in \Delta(\Theta)$, the designer's expected payoff is
\begin{equation}
V^0(\mathcal{D},\pi)
=
\int_{\Theta} u^0(\delta(\theta),\theta)\,\pi(d\theta).
\end{equation}
Analogous to \Cref{def:stable-vanila}, we define stable fictitious direct mechanisms as follows.
\begin{definition}
    \label{def:direct-self-confirming}
    An incentive compatible fictitious direct mechanism $(\mathcal{D},\varphi)$ is $\pi_0$-{stable} if there exists $\pi \in \Pi^{\mathcal{D}}(\varphi, \varphi_{\pi_0})$ such that \begin{equation} \label{eq:directselfconfirming} V^0(\mathcal{D},\pi) \geq V^0(\mathcal{D}',\pi)  \end{equation} for all incentive compatible direct mechanism $\mathcal{D}'$.
\end{definition}

The definition of stability is extended to apply to all \textit{incentive compatible} fictitious mechanisms. In applications, we should focus on those that are \textit{incentive
implementable}: by the fictitious revelation principle, these correspond exactly
to $\pi_0$-stable augmented mechanisms, up to outcome
and information equivalence.
The alternative direct mechanism $\mathcal D'$ need
only be incentive compatible.

The following monotonicity of stability of fictitious direct mechanisms is immediate from Definition \ref{def:direct-self-confirming}:

\begin{lemma}[Monotonicity]
    \label{l3+}
   Suppose $(\mathcal{D},\varphi)$ and $ (\mathcal{D},\tilde{\varphi})$ are two fictitious direct mechanisms such that $\tilde{\varphi} \succsim \varphi$.  Then for any $\pi_0\in \Delta(\Theta)$, $(\mathcal{D},\varphi)$ is  $\pi_0$-stable if $(\mathcal{D}, \tilde{\varphi})$ is  $\pi_0$-stable.%
\end{lemma}

\subsection{Refinement: A Grain of Truth}\label{sec:grain}

Stability offers a conceptual benchmark, but it is permissive. For instance, in the standard monopoly
setting with a buyer's valuation uniformly distributed
on $[0,1]$, every posted price can be
stable (the
maximin criterion in \Cref{rem:maxmin} does not help either).

The equilibrium approach to mechanisms opens up the possibility of refinements in the spirit of classical game theory. %
Although the message spaces and equilibrium strategies may differ substantially across mechanisms, what ultimately matters for the designer is the induced allocation and the information revealed, as captured by fictitious direct mechanisms.

The idea we exploit is what we term \textit{a  grain of truth}, which perturbs a mechanism so that it becomes marginally more informative in the sense that the true prior is revealed on a small set of types, where smallness is measured under the \textit{unknown} true prior. We interpret this as reflecting exogenous information leakage or acquisition.\footnote{The motivation is natural: a designer facing uncertainty may prefer to experiment locally before making substantial changes to the mechanism. For a related use of small experiments to identify local policy improvements, see \citet{ortner2018making}, who perturb the distribution of monitors' wages.} This perturbation consequently gives rise to a family of smaller sets of feasible priors. By the fictitious revelation principle, the perturbation can be carried out on the filter of the canonical fictitious direct mechanism.  %

For ease of presentation, we restrict attention to independent type distributions, both for the true prior and for the feasible priors; that is, $\pi_0=(\pi_0^1,\ldots,\pi_0^n) \in \prod_{i \in N} \Delta(\Theta^i)$, and for any fictitious direct mechanism $(\mathcal{D},\varphi)$,
\[
\Pi^{\mathcal{D}}(\varphi, \varphi_{\pi_0})
= \left\{\pi=(\pi^1,\ldots,\pi^n) \in \prod\nolimits_{i \in N} \Delta(\Theta^i)
: \varphi_{\pi}=\varphi_{\pi_0}\right\}.
\]
The formulation and analysis can be extended to general correlated distributions, but doing so would introduce substantially more notation. %
For any $i \in N$, given a set $A^i\subseteq \Theta^i$, 
\begin{equation*}
    \{\pi^i \in \Delta(\Theta^i): \pi^i|_{A^i}=\pi_0^i|_{A^i}\}
\end{equation*}
is the set of  priors on $\Theta^i$ that agree with $\pi_0^i$ on $A^i$. Now
\begin{equation}\label{eq:prior-selection}
            \Pi^{\mathcal{D}}(\varphi, \varphi_{\pi_0}) \cap \left\{\pi \in \prod\nolimits_{j \in N}\Delta(\Theta^j): \pi^i|_{A^i}=\pi_0^i|_{A^i}, i \in N\right\}
\end{equation}
selects feasible priors that agree with $\pi_0^i$ on $A^i$ for all $i$. %

The result below is conceptually important: it confirms that the intuitive refinement of feasible priors in \eqref{eq:prior-selection} is not arbitrary; it can be implemented through a well-defined filter. In particular, Parts \textup{(ii)} and \textup{(iii)} of
\Cref{lemma:perturbedfilter}, together with the fictitious
revelation principles in
\Cref{thm:stochastic-RP,thm:deterministic-RP},
ensure that the perturbation admits an incentive compatible
augmented implementation.

\begin{lemma}
\label{lemma:perturbedfilter}
For every product set
$A=\prod_{i\in N}A^i\subseteq\Theta$,
there exists a filter
$\varphi^A:\Theta\to\Delta(\Theta)$
of $\mathcal D$ such that
\begin{enumerate}
    \item[\textup{(i)}]
    The following holds:
    \begin{equation}
\label{eq:perturbedfilter}
\Pi^{\mathcal{D}}(\varphi^A, \varphi^A_{\pi_0})
=
\Pi^{\mathcal{D}}(\varphi, \varphi_{\pi_0}) \cap \left\{\pi \in \prod\nolimits_{j \in N}\Delta(\Theta^j): \pi^i|_{A^i}=\pi_0^i|_{A^i},  i \in N\right\}.
\end{equation}

    \item[\textup{(ii)}]
    If $\varphi$ is deterministic, then so is $\varphi^A$.

    \item[\textup{(iii)}]
    If $(\mathcal D,\varphi)$ is incentive implementable, %
    then so is $(\mathcal D,\varphi^A)$.
\end{enumerate}
\end{lemma}

With \Cref{lemma:perturbedfilter} in place, we can define the general $\varepsilon$-perturbation of any filter $\varphi$.

\begin{definition}
For $\varepsilon > 0$, we say that a filter $\varphi^\varepsilon$ \textbf{reveals $\varepsilon$-truth} about  $\varphi$ under $\pi_0$
if there exists 
$A = \prod_{i \in N} A^i \subseteq \Theta$ with 
$\pi_0^i(A^i) < \varepsilon$ for all $i\in N$ such that
\begin{equation*}
\Pi^{\mathcal{D}}(\varphi^\varepsilon, \varphi^\varepsilon_{\pi_0})
=
\Pi^{\mathcal{D}}(\varphi, \varphi_{\pi_0})
\cap\left\{\pi \in \prod\nolimits_{j \in N}\Delta(\Theta^j): \pi^i|_{A^i}=\pi_0^i|_{A^i},  i \in N\right\}.
\end{equation*}
\end{definition}

The refinement of stability requires that not only that $(\mathcal{D},\varphi)$ is $\pi_0$-stable but also that \textit{all} nearby fictitious direct mechanisms are $\pi_0$-stable.

\begin{definition}
\label{def:robust-stability}
An incentive compatible fictitious direct mechanism $(\mathcal{D},\varphi)$ is 
\textbf{robustly $\pi_0$-stable} if there exists $\varepsilon > 0$ such that 
$(\mathcal{D},\varphi^{\varepsilon})$ is $\pi_0$-stable for every filter 
$\varphi^{\varepsilon}$ that reveals $\varepsilon$-truth about $\varphi$ under $\pi_0$. 
\end{definition}

A feature of robust stability is that the
feasible prior establishing stability of
$(\mathcal D,\varphi^\varepsilon)$ may depend on
$\varphi^\varepsilon$, as allowed by
\Cref{def:direct-self-confirming}.
Each perturbation must admit a feasible prior under
which $\mathcal D$ is optimal, but this prior need
not be the same across perturbations. This is important.
With atomless true distributions, requiring
a single prior would force it to equal $\pi_0$, since
it must agree with $\pi_0$ on every sufficiently small set; consequently, the refinement would then impose optimality under the true prior by construction, rather than test
whether stability survives small informational
perturbations.

Since filters and fictitious direct mechanisms are technical devices, what we are after is actual incentive compatible allocations:
\begin{equation}
\mathscr D
:=
\left\{
\mathcal D:
\begin{array}{l}
(\mathcal D,\varphi)\text{ is incentive implementable and}\\
\text{robustly }\pi_0\text{-stable for some filter }\varphi
\end{array}
\right\}.
\end{equation}
We therefore introduce the following definition.
\begin{definition}
    If there exists a filter $\varphi$ such that the fictitious direct mechanism $(\mathcal{D},\varphi)$ is incentive implementable and robustly $\pi_0$-stable, then we say the direct mechanism $\mathcal{D}$ is \textbf{robustly $\pi_0$-stable}.
\end{definition}

When $(\mathcal D,\varphi)$ is incentive implementable,
it is natural to ask whether \Cref{def:robust-stability} should
also require its perturbations
$(\mathcal D,\varphi^\varepsilon)$ to be incentive implementable
when testing robust stability. This question concerns the agents' incentives in an actual mechanism (which is what implementation is about), and shouldn't be confused with the optimality requirement for the designer's choice of $\mathcal D$ with respect to the information
generated by $\varphi^\varepsilon$.
Our modeling choice is consistent with the fundamental choice in the usual equilibrium refinements.\footnote{For instance,
the completely mixed profiles used to define
trembling-hand perfection need not be equilibria of
the unperturbed game, and those used to establish
consistency in sequential equilibrium need not be
sequentially rational.
This perspective motivates allowing informational
perturbations without requiring incentive implementability in \Cref{def:robust-stability}.} However, the following result shows that this distinction is irrelevant: for an incentive implementable
fictitious direct mechanism, requiring stability under all
perturbations is equivalent to requiring stability
only under incentive implementable perturbations. This follows from \Cref{lemma:perturbedfilter}:  for every $\varphi^\varepsilon$ that reveals $\varepsilon$-truth about $\pi_0$ on a set $A$, there exists an information equivalent $\varphi^A$ such that
$(\mathcal D,\varphi^A)$ is incentive implementable. An analogous result holds for deterministic filters/mechanisms.

\begin{corollary}[Equivalent characterization of robust stability]
\label{cor:IC-robust-stability}
Consider a fictitious direct mechanism $(\mathcal D,\varphi)$.
\begin{enumerate}
    \item[\textup{(i)}]
    If $\varphi$ is  deterministic,
    then $(\mathcal D,\varphi)$ is robustly
    $\pi_0$-stable if and only if there exists
    $\varepsilon>0$ such that
    $(\mathcal D,\varphi^\varepsilon)$ is
    $\pi_0$-stable for every deterministic filter
    $\varphi^\varepsilon$ that reveals
    $\varepsilon$-truth about $\varphi$ under $\pi_0$.

    \item[\textup{(ii)}]
    If $(\mathcal D,\varphi)$ is incentive implementable,
    then it is robustly $\pi_0$-stable if and only if
    there exists $\varepsilon>0$ such that every
    incentive implementable fictitious direct mechanism
    $(\mathcal D,\varphi^\varepsilon)$ whose filter
    reveals $\varepsilon$-truth about $\varphi$
    under $\pi_0$ is $\pi_0$-stable.
\end{enumerate}
\end{corollary}

The monotonicity of robust stability follows from the monotonicity of stability, which is formalized by the following corollary. 

\begin{corollary}[Monotonicity]
\label{cor:monotonicrobustlysc}
    Suppose $(\mathcal{D},\varphi) \text{ and } (\mathcal{D},\tilde{\varphi})$ are fictitious direct mechanisms such that $\tilde{\varphi} \succsim \varphi $. Then $(\mathcal{D},\varphi)$ is  robustly $\pi_0$-stable if $(\mathcal{D}, \tilde{\varphi})$ is robustly $\pi_0$-stable.
\end{corollary}

The following consequence of
\Cref{cor:monotonicrobustlysc} provides a useful
verification criterion: robust stability can be checked
using a least informative filter, provided the associated
fictitious direct mechanism is incentive implementable.

\begin{theorem}[Verification of Robust Stability]
\label{thm:least-informative-filter}
Consider a direct mechanism $\mathcal D=(\Theta,\delta)$
and a filter $\varphi$ such that $\varphi\sim\delta$.
Suppose that $(\mathcal D,\varphi)$ is incentive implementable.
Then $\mathcal D$ is robustly $\pi_0$-stable if and only if
$(\mathcal D,\varphi)$ is robustly $\pi_0$-stable.
\end{theorem}
When $n=1$, a filter $\varphi \sim \delta$ always exists. With $n>1$, $\varphi =(\varphi^1,...,\varphi^n)$ is a joint function, while $\delta$ does not have to be.

\section{Application: Monopoly}\label{sec:application}
\label{sectionapplication}
In this section, we apply the theory and tools developed above to the simplest mechanism design environment. Consider a monopolistic seller who sells a single good to a buyer ($n=1$) whose private value of the good is $\theta \in \Theta=[0,1]$. The buyer's ex post payoff is $\theta-p$ if he obtains the good and ends up paying $p$; his payoff is $-p$ if doesn't obtain the object and ends up paying $p$. The outcome space is $O=\{0,1\} \times \mathbb{R_{+}}$. 

In a direct mechanism, the outcome function is $\delta:\Theta \to \Delta(O)=\Delta(\{0,1\} \times \mathbb{R_{+}})$.  Any incentive-compatible and individually rational mechanism is \textit{payoff} equivalent to a randomized posted-price mechanism: it induces the same \textit{allocation probability} and \textit{expected transfer} for every type.\footnote{Here we use the assumptions that transfers are nonnegative.} This payoff equivalence, however, need not preserve the information content of the mechanism and the distribution over outcomes may differ. By contrast, for deterministic mechanisms, incentive compatibility and individual rationality imply a cutoff rule. This rule is both outcome equivalent and (kernel) information equivalent to a deterministic posted-price mechanism.

Formally, a deterministic posting price mechanism with price $p$ has an outcome function:
\[
\delta^p(\theta) :=
\begin{cases}
\bm{1}_{(1,p)}, & \text{if } \theta \ge p,\\[6pt]
\bm{1}_{(0,0)}, & \text{if } \theta < p, 
\end{cases}
\]
where $\bm{1}_{(1,p)}$ and $\bm{1}_{(0,0)}$ are Dirac measures. It is assumed here that tie-breaking favors trade. Therefore, $\delta^{p}_{\pi_0}$ , the empirical observation under the true prior $\pi_0$, is given by \[\delta_{\pi_0}^p=\pi_0([p,1])\bm{1}_{(1,p)}+(1-\pi_0([p,1]))\bm{1}_{(0,0)}.\]
This means that the seller observes the probability of the buyer's types at or above $p$. We write  $\mathcal{D}^p=(\Theta,\delta^p)$ as the direct mechanism that represents the mechanism of posting price $p$. 

 A mechanism with randomized price $\tilde{p} \in \Delta([0,1])$ has an outcome function $\delta^{\tilde{p}}: \Theta \to \Delta(O)$ such that for any $\theta \in \Theta$ and any measurable $E \subseteq O$,
\begin{equation*}
\label{eq: randomizedpostedprice}
    \delta^{\tilde{p}}(\theta)(E):=\displaystyle{\int}_0^1 \delta^p(\theta)(E) \tilde{p}(dp).
\end{equation*}
We write $\mathcal{D}^{\tilde{p}}=(\Theta,\delta^{\tilde{p}})$ as the direct mechanism that represents the mechanism of random posting price $\tilde{p}$.  

Given the true prior $\pi_0$, the \textbf{revenue function} $R_0:[0,1]\rightarrow \mathbb{R}$ associated with $\pi_0$ is defined by $R_0(p)=p\pi_0([p,1])$. Our starting point of the analysis is the usual regularity condition
on $\pi_0$.
For our analysis, we generalize the regularity condition in the following direction: we say that $R_0$ is \textbf{single-peaked} if there exist
$0\leq\underline p\leq\overline p\leq1$ such that
$R_0$ is strictly increasing on $[0,\underline p]$,
constant on $[\underline p,\overline p]$, and strictly
decreasing on $[\overline p,1]$. If $R_0$ is single-peaked, then all of its local maximizers are global maximizers.

\begin{theorem}
\label{thm:generalmechanism}
Suppose that the true prior distribution $\pi_0$
is atomless with full support on $[0,1]$, and that
its revenue function $R_0$ is single-peaked.
Then an incentive compatible and individually rational
direct mechanism $\mathcal D$ is robustly
$\pi_0$-stable if and only if it is optimal under $\pi_0$.
\end{theorem}

A strict version of the usual regularity condition
is sufficient for single-peakedness: $\pi_0$ has a
continuous, strictly positive density on $[0,1]$,
and marginal revenue is strictly decreasing in the
probability of sale.
Under this condition, every robustly stable mechanism
is payoff-equivalent, $\pi_0$-almost surely, to posting
the unique optimal price under $\pi_0$.
This provides a justification of the familiar optimal mechanism of
\citet{myerson1981optimal} as an equilibrium mechanism, obtained here without
assuming that the seller knows $\pi_0$.

Notice that the seller, as a player in this broader game, could end up using the optimal mechanism under the true prior without knowing this fact or the true prior. This is to say, from an outside observer's perspective, if the true prior satisfies the regularity condition, we are sure that the only robust stable mechanisms are optimal under the true prior, whatever it is. This is a quite strong implication.

We now drop the single-peakedness (regularity) condition, so local
revenue maximizers need not be global maximizers.
We say that $R_0$ \textbf{does not oscillate infinitely often}
if there exists a finite partition
\[
0=p_0<p_1<\cdots<p_m=1
\]
such that $R_0$ is either weakly increasing or weakly
decreasing on each interval $[p_{k-1},p_k]$.
Flat segments are allowed and do not themselves
count as oscillations.

\begin{theorem}
\label{t3}
Suppose that $\pi_0$ is atomless with full support
on $[0,1]$ and that $R_0$ does not oscillate
infinitely often.
Then a randomized posted-price mechanism
$\mathcal{D}^{\tilde p}=(\Theta,\delta^{\tilde p})$
is robustly $\pi_0$-stable if and only if:
\begin{enumerate}
    \item[\textup{(i)}]
    $R_0$ is constant on $\operatorname{supp}(\tilde p)$;

    \item[\textup{(ii)}]
    every price in $\operatorname{supp}(\tilde p)$
    is a local maximizer of $R_0$.
\end{enumerate}
\end{theorem}

The result shows that mechanisms that are suboptimal
under the true prior can nevertheless be robustly stable,
providing a rationale for their persistence. The key reduction in the proof is to invoke
\Cref{thm:least-informative-filter} and work with a
least informative filter whose associated fictitious
mechanism is incentive implementable.
The main difficulty is to construct, for every
sufficiently small informational perturbation, a feasible
prior under which the mechanism remains optimal.

Since all deterministic incentive compatible and individually
rational mechanisms are posted-price mechanisms, the following corollary is immediate.

\begin{corollary}
Suppose that $\pi_0$ is atomless with full support
on $[0,1]$ and that $R_0$ does not oscillate
infinitely often.
Then a deterministic incentive compatible and
individually rational mechanism is robustly
$\pi_0$-stable if and only if it is a posted-price
mechanism $\mathcal D^p=(\Theta,\delta^p)$
for some price $p$ that is a local maximizer of $R_0$.
\end{corollary}

The necessity of local optimality has a simple intuition:
if a price  is not locally optimal,
revealing even a small amount of information near that
price allows the designer to profitably deviate.
Conversely, a locally optimal price remains stable under all sufficiently small informational perturbations.
The finite-oscillation assumption allows us to combine
the local arguments using a common positive bound
on the size of these perturbations.

\section{Conclusion}

This paper frames mechanism design as a fixed-point problem where the designer must learn from the information generated by the mechanism itself. We establish a \textit{fictitious revelation principle}, demonstrating that any incentive-compatible mechanism can be represented by a direct mechanism with a filter that garbles reports to preserve informational content. This framework allows us to characterize equilibrium mechanisms that remain optimal given the empirical evidence they produce.

To mitigate the permissiveness of stable beliefs, we introduce the \textit{grain of truth} refinement, which requires stability even if the true distribution is revealed on a set of small measure. In the single-buyer context, this refinement eliminates prices that are not locally revenue-maximizing, as local information would otherwise incentivize deviation. This result provides an endogenous justification for the Myerson mechanism. Future research may extend this logic to multi-agent environments, to settings with a monopolist who has private information, or to models that explicitly capture the dynamic experimentation required to reach these equilibria.

\appendix

\section{Proofs for \Cref{sec:FRP}: Fictitious Revelation Principle}
\label{appendixa}

\subsection{Proof of \Cref{thm:deterministic-RP}}

\begin{proof}
Throughout, we use our convention identifying deterministic
maps with their associated Dirac-valued stochastic mappings.

To prove \textup{(i)}, first consider an informationally deterministic fictitious direct
mechanism $(\mathcal D,\varphi)$.
Since $\varphi$ is Blackwell more informative than
$\delta$, there exists a Borel stochastic mapping
$g:\Theta\to\Delta(O)$ such that $\delta=g\circ\varphi$;
that is,
\[
\delta(\theta)(E)
=
\int_\Theta g(\theta')(E)\,
\varphi(\theta)(d\theta')
\]
for every $\theta\in\Theta$ and every measurable
$E\subseteq O$.
Since $\varphi$ is deterministic, this identity becomes
\[
\delta(\theta)(E)
=
\int_\Theta g(\theta')(E)\,
\bm{1}_{\varphi(\theta)}(d\theta')
=
g(\varphi(\theta))(E).
\]
Thus $\delta(\theta)=g(\varphi(\theta))$
for every $\theta\in\Theta$.

Set $M^i:=\Theta^i$ for each $i$, so that
$M:=\prod_{i\in N}M^i=\Theta$, and define $\omega:=g$.
Let $\mathcal M:=(M,\omega)$ and define the pure
strategies $\sigma^i:\Theta^i\to M^i$ by
$\sigma^i(\theta^i):=\varphi^i(\theta^i)$.
Then $(\mathcal M,\sigma)$ is strategically deterministic,
with $\omega\circ\sigma=\delta$ and $\sigma=\varphi$.
Hence it implements $(\mathcal D,\varphi)$.

Conversely, consider a strategically deterministic
augmented mechanism $(\mathcal M,\sigma)$.
For each $i$, choose a Borel isomorphism
$h^i:M^i\to\Theta^i$, and write $h=(h^i)_{i\in N}$.
Define
\[
\delta:=\omega\circ\sigma,\qquad
\varphi^i:=h^i\circ\sigma^i,\qquad
g:=\omega\circ h^{-1},
\]
and let $\mathcal D:=(\Theta,\delta)$.
Since $\sigma$ is a pure strategy profile,
$\varphi$ is deterministic.
Moreover, $\delta=g\circ\varphi$, so $\varphi$
is a filter of $\mathcal D$.
Furthermore,
\[
\varphi=h\circ\sigma,
\qquad
\sigma=h^{-1}\circ\varphi.
\]
Thus $\varphi$ and $\sigma$ are Blackwell equivalent,
and hence kernel equivalent.
Consequently, $(\mathcal D,\varphi)$ is an informationally
deterministic fictitious direct representation of
$(\mathcal M,\sigma)$.

\begin{remark}
    If $\delta$ is deterministic, then
$g(\varphi(\theta))$ is Dirac for every $\theta$.
Replacing every non-Dirac value of $g$ by a fixed
Dirac measure therefore preserves $\delta=g\circ\varphi$
and makes $g$ Dirac-valued everywhere.
This modification is Borel because the set of Dirac
measures is Borel.
With this choice, the constructions above give a deterministic implementation. Conversely, if $(\mathcal M,\sigma)$ is deterministic, then
$\delta=\omega\circ\sigma$ is also deterministic.
Since $\varphi$ is already deterministic,
$(\mathcal D,\varphi)$ is therefore deterministic. \qed
\end{remark}

\medskip
We now prove \noindent\textup{(ii)}. If $(\mathcal M,\sigma)$ is incentive compatible,
the direct mechanism constructed in \textup{(i)}
inherits its ex post or dominant-strategy incentive
compatibility by the standard revelation principle.
Hence $(\mathcal D,\varphi)$ is incentive compatible
in the corresponding sense.

For the reverse direction, let $(\mathcal D,\varphi)$
be incentive compatible.
Use the message spaces, pure strategies, and
Borel stochastic outcome function constructed
in the first direction of \textup{(i)}.
They satisfy
\[
M^i=\Theta^i,\qquad
\sigma^i=\varphi^i
\quad\text{for every }i\in N,
\qquad
\omega\circ\varphi=\delta.
\]
Write
\[
\hat M^i:=\varphi^i(\Theta^i),
\qquad
\hat M:=\prod_{i\in N}\hat M^i.
\]
Each $\hat M^i$ is analytic because it is the image
of a standard Borel space under a Borel map.
However, $\hat M^i$ need not be Borel.

For dominant-strategy incentive compatibility,
let $M'$ be the set of message profiles $m\in M$
for which there exist an agent $i$ and a true type
profile $\theta$ such that
\begin{equation}\label{eqn:deviation}
u^i(\omega(m),\theta)>
u^i\!\left(
\omega(\varphi^i(\theta^i),m^{-i}),\theta
\right).
\end{equation}
For ex post incentive compatibility, use the same
definition with the additional condition
\begin{equation}\label{eqn:ex-post-deviation}
m^{-i}=\varphi^{-i}(\theta^{-i}).
\end{equation}

By definition, the payoff expressions in \eqref{eqn:deviation}
are Borel functions in $(m,\theta)$.
The additional equality \eqref{eqn:ex-post-deviation}
in the ex post case also defines a Borel set.
Hence $M'$ is analytic, being a finite union of
projections of Borel subsets of $M\times\Theta$.
Moreover, $M'\cap\hat M=\emptyset$.
To see this, note that on $\hat M$, every unilateral
message deviation can be reproduced by a type report
in $\mathcal D$.
The corresponding incentive compatibility inequality
therefore rules out \eqref{eqn:deviation}.

Since $\hat M^i$ need not be Borel, we shall construct
Borel sets $\bar M^i\supseteq\hat M^i$ such that
\[
M'\cap\prod_{i\in N}\bar M^i=\emptyset.
\]
We prove this claim by induction.
Suppose that $\bar M^1,\ldots,\bar M^{k-1}$
have been chosen so that
\[
M'\cap
\left(
\prod_{j<k}\bar M^j
\times
\prod_{j\geq k}\hat M^j
\right)=\emptyset.
\]
For $k=1$, this is precisely $M'\cap\hat M=\emptyset$.
Project
\[
M'\cap
\left(
\prod_{j<k}\bar M^j
\times\Theta^k
\times
\prod_{j>k}\hat M^j
\right)
\]
onto $\Theta^k$.
This projection is analytic and, by the induction
hypothesis, disjoint from $\hat M^k$.
By Lusin's separation theorem
(\citealp[Theorem~14.7]{kechris1995classical}),
there exists a Borel set
$\bar M^k\supseteq\hat M^k$ that is disjoint from
this projection.
Consequently,
\[
M'\cap
\left(
\prod_{j\leq k}\bar M^j
\times
\prod_{j>k}\hat M^j
\right)=\emptyset,
\]
completing the induction.

Choose $\bar m^i\in\hat M^i$ and define the Borel maps
$\rho^i:M^i\to M^i$ by
\[
\rho^i(m^i):=
\begin{cases}
m^i, & m^i\in\bar M^i,\\
\bar m^i, & m^i\notin\bar M^i.
\end{cases}
\]
Write $\rho=(\rho^i)_{i\in N}$ and
$\bar M:=\prod_{i\in N}\bar M^i$.
Define
\[
\bar\omega:=\omega\circ\rho.
\]
This outcome function is a Borel stochastic kernel.
Let $\bar{\mathcal M}:=(M,\bar\omega)$.
Since $\rho\circ\varphi=\varphi$,
\[
\bar\omega\circ\sigma
=\omega\circ\rho\circ\varphi
=\omega\circ\varphi
=\delta
\quad\textup{and}\quad
\sigma=\varphi.
\]
Thus $(\bar{\mathcal M},\sigma)$ is a strategically
deterministic implementation of $(\mathcal D,\varphi)$. If $(\mathcal D,\varphi)$ is deterministic, then
$\omega$ was chosen Dirac-valued everywhere in
\textup{(i)}.
Since $\rho$ is deterministic,
$\bar\omega=\omega\circ\rho$ is also Dirac-valued
everywhere.
Hence $(\bar{\mathcal M},\sigma)$ is deterministic
as well.

Finally, note that
\[
\begin{aligned}
u^i\!\left(
\bar\omega(\varphi^i(\theta^i),m^{-i}),\theta
\right)
&=
u^i\!\left(
\omega(\rho^i(\varphi^i(\theta^i)),
       \rho^{-i}(m^{-i})),\theta
\right)\\
&=
u^i\!\left(
\omega(\varphi^i(\theta^i),
       \rho^{-i}(m^{-i})),\theta
\right)\\
&\geq u^i(\omega(\rho(m)),\theta)\\
&=u^i(\bar\omega(m),\theta)
\end{aligned}
\]
for every $i$, $\theta$, and $m$ in the
dominant-strategy case.
The second equality follows because $\rho^i$ fixes
every message in $\hat M^i$.
The weak inequality follows because $\rho(m)\in\bar M$
and $\bar M\cap M'=\emptyset$, so
\eqref{eqn:deviation} does not hold at $\rho(m)$.

In the ex post case, the same argument applies
whenever $m^{-i}=\varphi^{-i}(\theta^{-i})$,
because $\rho^{-i}$ fixes these prescribed messages.
These inequalities also rule out mixed deviations
by linearity.
Hence $(\bar{\mathcal M},\sigma)$ is incentive compatible
in the corresponding sense.
\end{proof}

\subsection{Proof of \Cref{thm:stochastic-RP}}
\begin{proof}

\noindent\textup{(i)} We first prove the implementation of a fictitious direct mechanism $(\mathcal{D}, \varphi)$. %
Since $\varphi$ is Blackwell more informative than $\delta$, there exists a measurable mapping $g : \Theta \to \Delta(O)$ such that $\delta = g \circ \varphi$. 
Consider the augmented mechanism $(\mathcal{M}, \sigma)$ defined by $\omega := g$, $M^i: = \Theta^i$, and $\sigma^i: = \varphi^i$ for each $i \in N$. Then $\delta = \omega \circ \sigma$ and $\sigma \sim \varphi$. Hence, $(\mathcal{M}, \sigma)$ is equivalent to $(\mathcal{D}, \varphi)$.

We now prove the representation of any augmented  mechanism $(\mathcal{M}, \sigma)$. %
Consider a direct mechanism $\mathcal{D}$ with outcome function 
$\delta := \omega \circ \sigma$.
We seek a filter $\varphi : \Theta \to \Delta(\Theta)$ such that 
$\varphi \sim \sigma$ and $\varphi$ is Blackwell more informative than $\delta$.
Since for each $i$ the type space $\Theta^i$ and the message space $M^i$ 
are uncountable standard Borel spaces, they are Borel isomorphic. 
Hence, without loss of generality, we may identify $M^i$ with $\Theta^i$ 
and write agent $i$’s strategy as 
$\sigma^i : \Theta^i \to \Delta(\Theta^i).$
Define $\varphi := \sigma$. Then, by construction, $\varphi \sim \sigma$, 
and hence
$\delta = \omega \circ \sigma = \omega \circ \varphi.$
Therefore, $\varphi$ is Blackwell more informative than $\delta$. Therefore, $(\mathcal{D}, \varphi)$ is a fictitious direct mechanism equivalent to $(\mathcal{M}, \sigma)$.

\medskip
\noindent\textup{(ii)}
The representation constructed in part \textup{(i)}
has $\delta=\omega\circ\sigma$.
By the standard revelation principle, $\mathcal D$
inherits ex post or dominant-strategy incentive compatibility
from $(\mathcal M,\sigma)$.
Hence $(\mathcal D,\varphi)$ is incentive compatible
in the corresponding sense.
\Cref{ex:ic-converse} in Online Appendix \ref{sec:incentive-implementability} establishes
that the converse need not hold.
\end{proof}

\section{Proofs for \Cref{sec:solution-concept}: Stable Mechanisms}

 Since agents' type distributions are independent, we have
    \begin{align*}
\left\{\pi \in \prod\nolimits_{i \in N}\Delta(\Theta^i): \sigma_{\pi}=\sigma_{\pi_0} \right\}
&= \prod\nolimits_{i\in N}\left\{\pi^i \in \Delta(\Theta^i): \sigma^i_{\pi^i}=\sigma^i_{\pi_0^i}\right\}, \\
\left\{\pi \in \prod\nolimits_{i \in N}\Delta(\Theta^i): \varphi_{\pi}=\varphi_{\pi_0} \right\}
&= \prod\nolimits_{i\in N}\left\{\pi^i \in \Delta(\Theta^i): \varphi^i_{\pi^i}=\varphi^i_{\pi_0^i}\right\}.
\end{align*}
Therefore the following holds: 
\begin{lemma}
    \label{lemma:independent}
    Consider any strategy profile $\sigma=(\sigma^i)_{i\in N}$ and any filter $\varphi=(\varphi^i)_{i\in N}$. If the prior distributions of agents' types are independent, then $\sigma \sim \varphi$ if and only if $\sigma^i \sim \varphi^i$ for all $i\in N$.
\end{lemma}

\subsection{Proof Sketch for  \Cref{lemma:perturbedfilter}}
The proof of \Cref{lemma:perturbedfilter} is relegated to the Online Appendix \ref{sec:lemma:perturbedfilter}. We sketch the proof strategy here. 
First, for any every product set
$A=\prod_{i\in N}A^i\subseteq\Theta$
with each $A^i$ Borel, and any fictitious direct mechanism $(\mathcal{D}, \varphi)$, we construct an augmented mechanism $(\mathcal{M},\sigma)$ such that 
\begin{gather}
M^i = \Theta^i \times \Theta^i, \nonumber \\
\sigma^i : \Theta^i \rightarrow \Delta(\Theta^i \times \Theta^i), \nonumber \\
\left\{\pi^i \in \Delta(\Theta^i): \sigma^i_{\pi^i}=\sigma^i_{\pi_0^i}\right\}
= \Pi^{\mathcal{D}}\left(\varphi^i, \varphi^i_{\pi_0^i}\right)
\cap
\left\{\pi^i \in \Delta(\Theta^i): \pi^i|_{A^i}=\pi_0^i|_{A^i}\right\},
\label{eq:sigmai}
\end{gather}
where the existence of an outcome function $\omega$ such that $\omega \circ \sigma=\delta$ will be argued and shown. We then invoke the fictitious revelation principle on the augmented mechanism to show the existence of a fictitious direct mechanism with a filter  $\varphi^A=(\varphi^{A^1},...,\varphi^{A^n})\sim \sigma$, i.e.,
\begin{equation} \label{eq: phiasigma} \Pi^{\mathcal{D}}(\varphi^A, \varphi^A_{\pi_0})=\left\{\pi \in \prod\nolimits_{i \in N} \Delta(\Theta^i): \sigma_{\pi}=\sigma_{\pi_0} \right\}, \end{equation}
where $\varphi^{A^i}:\Theta^i \rightarrow \Delta(\Theta^i)$ and, by Lemma \ref{lemma:independent}, $\varphi^{A^i} \sim \sigma^i$. By Lemma \ref{lemma:independent} and \eqref{eq:sigmai}, 
\begin{equation}
\label{eq:sigmaindependence}
\begin{aligned}
\left\{\pi \in \prod_{i \in N} \Delta(\Theta^i) : \sigma_{\pi} = \sigma_{\pi_0}\right\}
&= \prod_{i \in N} \left(\Pi^{\mathcal{D}} (\varphi^i, \varphi^i_{\pi_0^i})
 \cap 
\left\{\pi^i \in \Delta(\Theta^i) : \pi^i|_{A^i} = \pi_0^i|_{A^i}\right\} \right) \\
&= \Pi^{\mathcal{D}}(\varphi, \varphi_{\pi_0})
 \cap 
\left\{\pi \in \prod_{i \in N} \Delta(\Theta^i) : \pi^i|_{A^i} = \pi_0^i|_{A^i}, \; i\in N\right\}.
\end{aligned}
\end{equation}
Then \eqref{eq:perturbedfilter} follows from \eqref{eq: phiasigma} and \eqref{eq:sigmaindependence}. 
Therefore, to establish part~\textup{(i)},
it suffices to construct $\sigma^i$ satisfying
\eqref{eq:sigmai}.

We establish \textup{(ii)} by choosing the fictitious
direct mechanism representation through Borel relabeling,
which preserves pure strategies. For \textup{(iii)}, we apply the same construction
to an incentive compatible implementation of $(\mathcal D,\varphi)$
and show that the resulting incentive compatible augmented mechanism
implements the same filter $\varphi^A$.

\subsection{Proof of \Cref{cor:IC-robust-stability}}
\begin{proof}
The ``only if'' directions for both cases are immediate. For either converse, fix $\varepsilon>0$ satisfying
the stated condition, and consider any filter
$\varphi^\varepsilon$ revealing $\varepsilon$-truth
about $\varphi$ under $\pi_0$.
Let $A$ be the set of type profiles that are revealed. 
By \Cref{lemma:perturbedfilter}, there exists a filter
$\varphi^A$ satisfying
\[
\Pi^{\mathcal D}(\varphi^A,\varphi^A_{\pi_0})
=
\Pi^{\mathcal D}
(\varphi^\varepsilon,\varphi^\varepsilon_{\pi_0}).
\]
Moreover, $\varphi^A$ is deterministic in case
\textup{(i)}, and $(\mathcal D,\varphi^A)$ is
incentive implementable in case \textup{(ii)}.
Thus $(\mathcal D,\varphi^A)$ is $\pi_0$-stable
by assumption.
The same supporting prior establishes
$\pi_0$-stability of $(\mathcal D,\varphi^\varepsilon)$.
Since $\varphi^\varepsilon$ was arbitrary,
$(\mathcal D,\varphi)$ is robustly $\pi_0$-stable.
\end{proof}

\subsection{Proof of Corollary \ref{cor:monotonicrobustlysc}}

\begin{proof} 
If $(\mathcal{D}, \tilde{\varphi})$ is robustly $\pi_0$-stable, then there exists $\varepsilon>0$ such that 
$(\mathcal{D},\tilde{\varphi}^{\varepsilon})$ is $\pi_0$-stable for every filter 
$\tilde{\varphi}^{\varepsilon}$ that reveals $\varepsilon$-truth about $\tilde{\varphi}$ under $\pi_0$. We show that $(\mathcal{D},\varphi)$ is robustly $\pi_0$-stable with the same $\varepsilon$.

Consider any filter $\varphi^\varepsilon$ that reveals $\varepsilon$-truth about $\varphi$ under $\pi_0$: there exists $A=\prod_{i \in N}A^i\subseteq \Theta$ with $\pi_0^i(A^i)< \varepsilon$ for all $i\in N$ such that
\begin{equation}\label{eq:phi-eps}
    \Pi^{\mathcal{D}}(\varphi^\varepsilon, \varphi^\varepsilon_{\pi_0})
    =\Pi^{\mathcal{D}}(\varphi, \varphi_{\pi_0}) \cap \left\{\pi \in \prod\nolimits_{j \in N}\Delta(\Theta^j): \pi^i|_{A^i}=\pi_0^i|_{A^i},  i \in N\right\}.
\end{equation}
By assumption $\tilde{\varphi} \succsim \varphi$, $\Pi^{\mathcal{D}}(\tilde{\varphi}, \tilde{\varphi}_{\pi_0}) \subseteq \Pi^{\mathcal{D}}(\varphi, \varphi_{\pi_0})$. By \eqref{eq:phi-eps},
\begin{equation}\label{eq:ranking}
    \Pi^{\mathcal{D}}(\varphi^\varepsilon, \varphi^\varepsilon_{\pi_0}) \supseteq \Pi^{\mathcal{D}}(\tilde{\varphi},\tilde{\varphi}_{\pi_0}) \cap \left\{\pi \in \prod\nolimits_{j \in N}\Delta(\Theta^j): \pi^i|_{A^i}=\pi_0^i|_{A^i},  i \in N\right\}.
\end{equation}
By \Cref{lemma:perturbedfilter} applied to $(\mathcal{D},\tilde{\varphi})$ and $A$, there exists a filter $\tilde{\varphi}^A$ of $\mathcal{D}$ such that
\begin{equation}\label{eq:represent_A}
    \Pi^{\mathcal{D}}(\tilde{\varphi}^A, \tilde{\varphi}^A_{\pi_0})=\Pi^{\mathcal{D}}(\tilde{\varphi},\tilde{\varphi}_{\pi_0}) \cap \left\{\pi \in \prod\nolimits_{j \in N}\Delta(\Theta^j): \pi^i|_{A^i}=\pi_0^i|_{A^i},  i \in N\right\}.
\end{equation}
Since $\pi_0^i(A^i)<\varepsilon$ for all $i\in N$, \eqref{eq:represent_A} says precisely that $\tilde{\varphi}^A$ reveals $\varepsilon$-truth about $\tilde{\varphi}$ under $\pi_0$. By the choice of $\varepsilon$, $(\mathcal{D},\tilde{\varphi}^A)$ is $\pi_0$-stable: there exists $\pi^*\in \Pi^{\mathcal{D}}(\tilde{\varphi}^A, \tilde{\varphi}^A_{\pi_0})$ such that
\begin{equation}\label{eq:pistar-optimal}
    V^0(\mathcal{D},\pi^*)\geq V^0(\mathcal{D}',\pi^*)
\end{equation}
for all incentive compatible direct mechanisms $\mathcal{D}'$. By \eqref{eq:represent_A} and \eqref{eq:ranking}, $\pi^* \in \Pi^{\mathcal{D}}(\varphi^\varepsilon, \varphi^\varepsilon_{\pi_0})$. Together with \eqref{eq:pistar-optimal}, this shows that $(\mathcal{D},\varphi^\varepsilon)$ is $\pi_0$-stable.

Since $\varphi^\varepsilon$ was an arbitrary filter that reveals $\varepsilon$-truth about $\varphi$ under $\pi_0$, $(\mathcal{D},\varphi)$ is robustly $\pi_0$-stable.
\end{proof}

\subsection{Proof of \Cref{thm:least-informative-filter}}

\begin{proof}
If $(\mathcal D,\varphi)$ is robustly $\pi_0$-stable,
then, since it is incentive implementable by assumption,
$\mathcal D$ is robustly $\pi_0$-stable by definition.
Conversely, suppose that $\mathcal D$ is robustly
$\pi_0$-stable.
By definition, there exists a filter $\varphi^{\mathcal D}$
such that $(\mathcal D,\varphi^{\mathcal D})$ is
incentive implementable and robustly $\pi_0$-stable.
Since $\varphi^{\mathcal D}$ is a filter of $\mathcal D$
and $\varphi\sim\delta$, we have
\(
\varphi^{\mathcal D}\succsim\delta\sim\varphi.
\)
Hence, by \Cref{cor:monotonicrobustlysc},
$(\mathcal D,\varphi)$ is robustly $\pi_0$-stable.
\end{proof}

When $n=1$, consider the augmented mechanism $(O, \id_O,\delta)$, where $\id_O$ is the identity map on $O$. By  the fictitious revelation principle (\Cref{thm:stochastic-RP} (i)), it has a representation of a fictitious direct mechanism $(\mathcal{D},\varphi)$, where $\varphi \sim \delta$.

\section{Proofs for \Cref{sec:application}: Monopoly Application}

\subsection{Proof of Theorem \ref{thm:generalmechanism}}

\begin{proof}
Suppose first that $\mathcal D$ is optimal under $\pi_0$.
Let $\varphi^*(\theta):=\bm{1}_\theta$ reveal the type
perfectly.
Its feasible-prior set is $\{\pi_0\}$, and every
$\varepsilon$-truth perturbation leaves this set unchanged.
Hence $(\mathcal D,\varphi^*)$ is robustly $\pi_0$-stable,
and the truthful direct mechanism provides an
incentive compatible implementation.

Conversely, suppose that $\mathcal D$ is robustly
$\pi_0$-stable. By definition, there exists a filter
$\varphi$ such that $(\mathcal D,\varphi)$ is
incentive implementable and robustly $\pi_0$-stable.
Choose $\varepsilon>0$ as in
\Cref{def:robust-stability}. By \Cref{lemma:perturbedfilter}, for every Borel set
$A\subseteq[0,1]$ with $\pi_0(A)<\varepsilon$,
there exists a prior $\pi_A$ under which $\mathcal D$
is optimal and
\[
\varphi_{\pi_A}=\varphi_{\pi_0},
\qquad
\pi_A|_A=\pi_0|_A.
\]
Since $\delta=g\circ\varphi$ for some $g$, we also have 
\begin{equation}
\label{eqn:outcome-equivalence}
\delta_{\pi_A}=\delta_{\pi_0}.
\end{equation}
In particular, expected allocation and transfer are
the same under $\pi_A$ and $\pi_0$.

Since $\pi_0$ is atomless, choose $b>0$ and a Borel
set $A\subseteq[b,1]$ with
$0<\pi_0(A)<\varepsilon$.
Optimality under $\pi_A$ gives
\[
V^0(\mathcal D,\pi_0)
=V^0(\mathcal D,\pi_A)
\ge b\pi_A([b,1])
\ge b\pi_0(A)>0.
\]

Write $q(\theta)$ and $t(\theta)$ for the expected
allocation and transfer under $\delta(\theta)$.
Incentive compatibility, individual rationality,
and nonnegative transfers imply that $q$ is
nondecreasing and
\[
t(\theta)
=\theta q(\theta)-\int_0^\theta q(s)\,ds.
\]
Define $\bar q(\theta):=q(\theta+)$ for $\theta<1$
and $\bar q(1):=1$, and let
$\widetilde p\in\Delta([0,1])$ have cumulative
distribution function $\bar q$.
The associated randomized posted-price mechanism
has expected allocation $\bar q$ and expected transfer
\[
\bar t(\theta)
=\int_{[0,\theta]}p\,\widetilde p(dp)
=t(\theta)+\theta\bigl(\bar q(\theta)-q(\theta)\bigr)
\ge t(\theta).
\]
Because $\bar q$ and $q$ differ at only countably
many points and $\pi_0$ is atomless,
\begin{equation}
\label{eq:expectedat}
\begin{aligned}
V^0(\mathcal D,\pi_0)
&=\int_0^1 R_0(p)\,\widetilde p(dp),\\
\int_0^1 q(\theta)\,\pi_0(d\theta)
&=\int_0^1\pi_0([p,1])\,\widetilde p(dp).
\end{aligned}
\end{equation}

For every prior $\pi_A$ obtained above, write
$R_A(p):=p\pi_A([p,1])$.
Since every posted-price mechanism is IC and IR,
optimality of $\mathcal D$ under $\pi_A$ implies
\[
R_A(p)\le V^0(\mathcal D,\pi_A)
=V^0(\mathcal D,\pi_0)
\qquad\text{for every }p.
\]
Together with $\bar t\ge t$ and
\eqref{eqn:outcome-equivalence}, this gives
\[
V^0(\mathcal D,\pi_0)
=\int t\,d\pi_A
\le\int\bar t\,d\pi_A
=\int R_A(p)\,\widetilde p(dp)
\le V^0(\mathcal D,\pi_0).
\]
Therefore,
\begin{equation}
\label{eq:samerevenue}
R_A(p)=V^0(\mathcal D,\pi_0)
\quad\text{for $\widetilde p$-almost every }p,
\end{equation}
and
\[
\int\theta(\bar q-q)\,d\pi_A=0.
\]
Since $R_A(p)\le p$, we also obtain
\[
\operatorname{supp}\widetilde p
\subseteq[V^0(\mathcal D,\pi_0),1].
\]
Thus $\bar q(0)=q(0)=0$, and the preceding integral
identity implies $\bar q=q$ $\pi_A$-almost surely.
Consequently,
\[
\begin{aligned}
\int\pi_A([p,1])\,\widetilde p(dp)
&=\int\bar q\,d\pi_A
=\int q\,d\pi_A\\
&=\int q\,d\pi_0
=\int\pi_0([p,1])\,\widetilde p(dp).
\end{aligned}
\]
For
\[
d(p):=\frac{V^0(\mathcal D,\pi_0)}p-\pi_0([p,1]),
\qquad p\in(0,1],
\]
it follows from \eqref{eq:expectedat} and
\eqref{eq:samerevenue} that
\begin{equation}
\label{eq:atdifference}
\int_{(0,1]}d(p)\,\widetilde p(dp)=0
\quad \textup{and} \quad
\int_{(0,1]}p\,d(p)\,\widetilde p(dp)=0.
\end{equation}
Both integrals are well defined because the support
of $\widetilde p$ is bounded away from zero.

Fix $p\in\operatorname{supp}\widetilde p$.
By atomlessness, choose a relatively open interval
$A\subseteq[0,1]$ containing $p$ with
$\pi_0(A)<\varepsilon$.
Since $\pi_A|_A=\pi_0|_A$, the function $R_A$
is continuous on $A$.
By \eqref{eq:samerevenue}, the bound
$R_A\le V^0(\mathcal D,\pi_0)$, and
$p\in\operatorname{supp}\widetilde p$, we have
\[
R_A(p)=V^0(\mathcal D,\pi_0).
\]
In particular, $p=1$ is impossible because
$R_A(1)=\pi_A(\{1\})=0$.
Hence $\operatorname{supp}\widetilde p\subseteq(0,1)$.

For every $r\in A$,
\[
\pi_A([r,1])-\pi_A([p,1])
=\pi_0([r,1])-\pi_0([p,1]).
\]
Using $R_A(p)=V^0(\mathcal D,\pi_0)$, we obtain
\[
R_A(r)-R_A(p)
=R_0(r)-R_0(p)+(r-p)d(p)\le0.
\]

Let $P^*$ be the set of maximizers of $R_0$,
and choose $p^*\in P^*$.
If $p<p^*$, choose $r\in A\cap(p,p^*)$.
Single-peakedness gives $R_0(r)\ge R_0(p)$,
strictly if $p\notin P^*$, so $d(p)\le0$,
strictly if $p\notin P^*$.
Similarly, if $p>p^*$, choosing
$r\in A\cap(p^*,p)$ gives $d(p)\ge0$,
again strictly if $p\notin P^*$.
Thus
\[
(p-p^*)d(p)\ge0
\quad\text{on }\operatorname{supp}\widetilde p,
\]
with strict inequality outside $P^*$.
But \eqref{eq:atdifference} gives
\[
\int_{(0,1]}(p-p^*)d(p)\,\widetilde p(dp)=0.
\]
Hence $\widetilde p(P^*)=1$, and
\[
V^0(\mathcal D,\pi_0)
=\int_0^1R_0(p)\,\widetilde p(dp)
=\max_{p\in[0,1]}R_0(p).
\]

For any other IC--IR mechanism $\mathcal D'$,
the same construction gives a price distribution
$\widetilde p'$ such that
\[
V^0(\mathcal D',\pi_0)
\le\int_0^1R_0(p)\,\widetilde p'(dp)
\le\max_{p\in[0,1]}R_0(p)
=V^0(\mathcal D,\pi_0).
\]
Therefore, $\mathcal D$ is optimal under $\pi_0$.
\end{proof}

\subsection{Proof of Theorem \ref{t3}}

\begin{proof}
    \textbf{Step 1: } \label{t3step1}We start with the claim that
    \begin{equation}
\label{b5}
\{\pi\in\Delta(\Theta):\delta^{\tilde p}_\pi=\delta^{\tilde p}_{\pi_0}\}
=\bigl\{\pi\in\Delta(\Theta):\pi([p,1])=\pi_0([p,1])\ \text{ for }\tilde p\text{-a.e.\ }p\bigr\}.
\end{equation}
 First by Fubini's Theorem, for any $\pi \in \Delta(\Theta), \text{ measurable } E \subseteq O, $ 
\begin{align*}
\delta^{\tilde p}_\pi(E)
&= \int_0^1 \delta^{\tilde p}(\theta)(E)\,\pi(d\theta)\\
&= \int_0^1 \int_0^1
\delta^p(\theta)(E)\,\tilde p(dp)\,\pi(d\theta)\\
&= \int_0^1 \int_0^1
\delta^p(\theta)(E)\,\pi(d\theta)\,\tilde p(dp)\\
&= \int_0^1 \delta^p_\pi(E)\,\tilde p(dp).
\end{align*}

To show (\ref{b5}), we first show the following lemma.
\begin{lemma}
\label{l5}
    For any $p \in [0,1],$ \begin{equation*} \{\pi \in \Delta(\Theta): \delta^p_{\pi}=\delta^p_{\pi_0}\}=\{\pi \in \Delta(\Theta): \pi([p,1])=\pi_0([p,1])\}, \end{equation*} where $\delta^p$ is the outcome function of posting price $p$.
\end{lemma}

\begin{proof}
    Recall that for any $p \in [0,1]$, \[
\delta^p(\theta) :=
\begin{cases}
1_{(1,p)}, & \text{if } \theta \ge p,\\[6pt]
1_{(0,0)}, & \text{if } \theta < p.
\end{cases}
\] Hence, \[\delta_{\pi}^p=\pi([p,1])1_{(1,p)}+(1-\pi([p,1]))1_{(0,0)}.\]
Therefore, if $\pi([p,1])=\pi_0([p,1])$, then $\delta_{\pi}^p=\delta_{\pi_0}^p.$ Conversely, if  $\delta_{\pi}^p=\delta_{\pi_0}^p,$ consider the outcome $(1,p)\in O$.  We have $\pi([p,1])=\delta^p_\pi(\{(1,p)\})=\delta^p_{\pi_0}(\{(1,p)\})=\pi_0([p,1])$. 
\end{proof}

To see the ``$\supseteq$'' direction of (\ref{b5}), suppose $\pi\in \Delta(\Theta)$ is such that $\pi([p,1])=\pi_0([p,1]), \tilde{p}\text{-a.e.}$, then by \Cref{l5}, $\delta^p_\pi=\delta^p_{\pi_0}, \tilde{p}\text{-a.e.}$ Therefore, \begin{equation*}
\delta^{\tilde{p}}_\pi(E)
= \displaystyle \int_0^1 \delta^p_\pi(E)\,\tilde{p}(dp)=\displaystyle \int_0^1 \delta^p_{\pi_0}(E)\,\tilde{p}(dp)=\delta^{\tilde{p}}_{\pi_0}(E).
\end{equation*}
To see the ``$\subseteq$'' direction of (\ref{b5}), for any measurable $B \subseteq [0,1]$,
\begin{align}
\int_{B} \pi([p,1])\,\tilde{p}(dp)
&= \int_{0}^{1} \pi([p,1])\, 1_{p}(B)\,\tilde{p}(dp)  \notag\\
&= \int_{0}^{1} \pi([p,1])\, 1_{(1,p)}(\{1\}\times B)\,\tilde{p}(dp) \notag\\
&= \int_{0}^{1} \left( \pi([p,1])\, 1_{(1,p)}(\{1\}\times B)
      + (1-\pi([p,1]))\, 1_{(0,0)}(\{1\}\times B) \right)\tilde{p}(dp) \notag\\
&= \int_{0}^{1} \delta^p_{\pi}(\{1\}\times B)\,\tilde{p}(dp).
\label{b7}
\end{align}
Similarly, \begin{equation}
\label{b7-beta}
\begin{aligned}
\int_{B} \pi_0([p,1])\,\tilde{p}(dp)= \int_{0}^{1} \delta^p_{\pi_0}(\{1\}\times B)\,\tilde{p}(dp).
\end{aligned}
\end{equation}
Since $\delta_{\pi}^{\tilde{p}}=\delta_{\pi_0}^{\tilde{p}}$, \eqref{b7} and \eqref{b7-beta} imply that $\int_{B} \pi([p,1])\,\tilde{p}(dp)=\int_{B} \pi_0([p,1])\,\tilde{p}(dp)$ for any measurable $B \subseteq [0,1]$ and hence $\pi([p,1])=\pi_0([p,1])$, $\tilde{p}$-a.e. 

In addition, the following lemma demonstrates what happens when $\pi([p,1]) \neq \pi_0([p,1])$:
\begin{lemma}
\label{lemma:equaltail}
Consider $\pi_0$ atomless. For any $\pi$ that belongs to the right-hand side of \eqref{b5}, and for any $p \in \supp \tilde{p}$, 
\begin{equation*}
    \label{eq:equaltail}
    \pi([p,1])\geq \pi_0([p,1]).
\end{equation*}
\end{lemma}

\begin{proof}
   The statement is true if $\tilde{p}(\{p\})>0$, in particular, $\pi([p,1])=\pi_0([p,1])$ when $\tilde{p}(\{p\})>0$.  
   
   If $\tilde{p}(\{p\})=0$, then there exists $\{p^n\}_{n \in \mathbb{N}} \subseteq \supp{\tilde{p}}$ such that 
   \begin{equation*} p^n \rightarrow p \text{ and } \pi([p^n,1])=\pi_0([p^n,1]), \forall n. \end{equation*}
   If $p^n<p$ for all large $n$ (take a subsequence if necessary), then $\pi([p^n,1]) \rightarrow \pi([p,1])$ and $\pi_0([p^n,1]) \rightarrow \pi_0([p,1])$ so $\pi([p,1]) = \pi_0([p,1])$. If $p^n>p$ for all large $n$ (take a subsequence if necessary), then $\pi([p^n,1]) \rightarrow \pi((p,1])$ and $\pi_0([p^n,1]) \rightarrow \pi_0([p,1])$. Therefore, $\pi([p,1])\geq \pi((p,1])=\pi_0((p,1])=\pi_0([p,1])$, as $\pi_0$ is atomless. 
   \end{proof} 

\begin{lemma}
\label{lem:postedpriceic}
For every $\tilde p\in\Delta([0,1])$, the pure mapping
\[
\theta\mapsto\tilde p([0,\theta])
\]
is kernel equivalent to $\delta^{\tilde p}$.
Moreover, $(\mathcal D^{\tilde p},\varphi)$ is
incentive implementable for every filter
$\varphi\sim\delta^{\tilde p}$.
\end{lemma}

\begin{proof}
Set $M=[0,1]$ and $\sigma(\theta):=\tilde p([0,\theta])$.
Define
\[
r(m):=\inf\{x\in[0,1]:\sigma(x)\ge m\},
\qquad
\omega(m):=\delta^{\tilde p}(r(m)).
\]
The quantile property gives
$r(\sigma(\theta))\le\theta$ and
$\sigma(r(\sigma(\theta)))=\sigma(\theta)$.
Hence no price mass lies between these two reports,
so $\omega\circ\sigma=\delta^{\tilde p}$.
Every message deviation $m$ corresponds to reporting
$r(m)$ in $\mathcal D^{\tilde p}$.
Thus $(M,\omega,\sigma)$ is incentive compatible.

Outcome equivalence implies that
$\sigma_\pi=\sigma_{\pi'}$ entails
$\delta^{\tilde p}_\pi=\delta^{\tilde p}_{\pi'}$.
Conversely, if
$\delta^{\tilde p}_\pi=\delta^{\tilde p}_{\pi'}$,
the argument establishing \eqref{b5} gives
\[
\pi([p,1])=\pi'([p,1])
\quad\text{for $\tilde p$-almost every }p.
\]
Since $r$ maps the uniform distribution on $(0,1)$
to $\tilde p$, we obtain
\[
\sigma_\pi([m,1])
=\pi([r(m),1])
=\pi'([r(m),1])
=\sigma_{\pi'}([m,1])
\]
for Lebesgue-almost every $m\in(0,1)$.
Left-continuity extends
equality to every $m\in(0,1]$; at $m=0$ both equal one.
Therefore $\sigma\sim\delta^{\tilde p}$.

For any filter $\varphi\sim\delta^{\tilde p}$,
we consequently have $\sigma\sim\varphi$.
The same incentive compatible augmented mechanism
therefore implements $(\mathcal D^{\tilde p},\varphi)$.
\end{proof}

The key reduction for the proof is to focus on a least informative
filter as established by \Cref{thm:least-informative-filter}. To this end,
we first verify that the corresponding fictitious
direct mechanism is incentive implementable.
Fix a filter $\varphi^{\tilde p}\sim\delta^{\tilde p}$.
The incentive compatible augmented mechanism constructed
in \Cref{lem:postedpriceic} satisfies
\[
\omega\circ\sigma=\delta^{\tilde p},
\qquad
\sigma\sim\delta^{\tilde p}\sim\varphi^{\tilde p}.
\]
Hence $(\mathcal D^{\tilde p},\varphi^{\tilde p})$
is incentive implementable.
By \Cref{thm:least-informative-filter},
$\mathcal D^{\tilde p}$ is robustly $\pi_0$-stable
if and only if
$(\mathcal D^{\tilde p},\varphi^{\tilde p})$
is robustly $\pi_0$-stable.

\textbf{Step 2: } We start with the part (i) of the ``only if'' direction. Given $\pi_0 \in \Delta(\Theta)$, we want to show that if there exists $p_1,p_2 \in \supp \tilde{p}$ such that $p_1 \cdot \pi_0([p_1,1]) \neq p_2 \cdot \pi_0([p_2,1])$, then $ \mathcal{D}^{\tilde{p}}$ is not robustly stable, which is equivalent to $( \mathcal{D}^{\tilde{p}}, \varphi^{\tilde{p}})$ is not robustly stable for any filter $\varphi^{\tilde{p}} \sim \delta^{\tilde{p}}$, i.e.,
    \begin{equation*}
    \label{b4}
        \Pi^{\mathcal{D}^{\tilde{p}}}(\varphi^{\tilde{p}}, \varphi^{\tilde{p}}_{\pi_0})=\{\pi \in \Delta(\Theta): \delta^{\tilde{p}}_{\pi}=\delta^{\tilde{p}}_{\pi_0}\}. 
    \end{equation*}

Without loss of generality, assume that $p_1 \cdot \pi_0([p_1,1])<p_2 \cdot \pi_0([p_2,1])$. Since $\pi_0$ is atomless, there exists $\alpha>0$ such that for any $p_1'\in(p_1-\alpha,p_1+\alpha)$, we have
\begin{equation} \label{thm1-3}  p_1' \pi_0([p_1',1])<p_2 \pi_0([p_2,1]). \end{equation} In addition, as $p_1 \in \supp \tilde{p}, $ \begin{equation} \label{thm1-4} \tilde{p}((p_1-\alpha,p_1+\alpha))>0. \end{equation} 
Hence, for any
$\pi\in\Pi^{\mathcal D^{\tilde p}}
(\varphi^{\tilde p},\varphi^{\tilde p}_{\pi_0})$,
\begin{align*}
V^0(\mathcal D^{\tilde p},\pi)
&= \int_0^1 p\,\pi([p,1])\,\tilde p(dp)\\
&= \int_0^1 p\,\pi_0([p,1])\,\tilde p(dp)\\
&< \max_{p\in\supp\tilde p} p\,\pi_0([p,1]),
\end{align*}
where the strict inequality is by (\ref{thm1-3}) and (\ref{thm1-4}), and the maximization is well-defined, as $\tilde{p}$ is a measure on a compact set, thus $\supp \tilde{p}$ is compact, and the mapping $p \mapsto p \cdot \pi_0([p,1])$ is continuous, as $\pi_0$ is atomless. 

Let $p^* \in \arg\max  \limits_{p \in \supp \tilde{p}} p \cdot \pi_0([p,1])$. By Lemma \ref{lemma:equaltail}, for any $\pi \in \Pi^{\mathcal{D}^{\tilde{p}}}(\varphi^{\tilde{p}}, \varphi^{\tilde{p}}_{\pi_0}),$ \begin{equation*} \pi([p^*,1]) \geq \pi_0([p^*,1]). \end{equation*}
Hence, 
\(V^0({\mathcal{D}^{\tilde{p}}},\pi)< p^*   \pi_0([p^*,1]) \leq p^*  \pi([p^*,1])=V({\mathcal{D}^{p^*}},\pi). \)
As a result, $(\mathcal{D}^{\tilde{p}}, \varphi^{\tilde{p}})$ is not stable, which is not robustly stable. 

\textbf{Step 3: } We proceed to part (ii) of the ``only if'' direction.  We want to show that if the first condition in \Cref{t3} is satisfied by $\tilde{p}$, but there exists $p_* \in \supp \tilde{p}$ such that $p_*$ is not a local maximizer of $p\pi_0([p,1])$, then $({\mathcal{D}^{\tilde{p}}}, \varphi^{\tilde{p}})$ is still not robustly stable. 

If $p_*$ is not a local maximizer of $p   \pi_0([p,1])$, then for any $\varepsilon>0$, there exists $\beta>0$ such that $\pi_0([p_*-\beta,p_*))<\varepsilon, \pi_0((p_*, p_*+\beta])<\varepsilon$, and either 
\begin{enumerate}
\item $p \pi_0([p,1])>p_* \pi_0([p_*,1])$ for all $p \in [p_*-\beta,p_*)$, or  \label{case1}
\item $p \cdot \pi_0([p,1])>p_*\pi_0([p_*,1])$ for all $p \in (p_*, p_*+\beta].$ \label{case2}
\end{enumerate}

If case \ref{case1},  let $A=[p_*-\beta,p_*)$ and by \Cref{lemma:perturbedfilter}, there exists a filter $\varphi^A$ such that 
\[       \Pi^{\mathcal{D}^{\tilde{p}}}(\varphi^A, \varphi^A_{\pi_0}) =\Pi^{\mathcal{D}^{\tilde{p}}}(\varphi^{\tilde{p}}, \varphi^{\tilde{p}}_{\pi_0}) \cap \{\pi \in \Delta(\Theta): \pi|_{A}=\pi_0|_A\}. \]
By the first condition in \Cref{t3}, for any $\pi \in      \Pi^{\mathcal{D}^{\tilde{p}}}(\varphi^A, \varphi^A_{\pi_0}) $ and $ p \in [p_*-\beta, p_*)$, we have $\pi([p,p_*))=\pi_0([p,p_*))$ and
\begin{equation*}  
V^0(\mathcal{D}^{\tilde{p}},\pi)=V^0(\mathcal{D}^{\tilde{p}},\pi_0)=V^0(\mathcal{D}^{p_*},\pi_0)<V^0(\mathcal{D}^p,\pi_0) \leq V^0(\mathcal{D}^p,\pi), \end{equation*}
where the first equality is by \eqref{b5}, the second equality is by the first condition in \Cref{t3}, the first inequality is by assumption, and the second inequality follows the following argument: \begin{equation*}
\begin{array}{rcl}
\pi([p,1])
&=& \pi([p,p_*)) + \pi([p_*,1]) \\[0.6em]
&\ge& \pi_0([p,p_*)) + \pi_0([p_*,1]) \\[0.6em]
&=& \pi_0([p,1]).
\end{array}
\end{equation*}
Therefore, $({\mathcal{D}^{\tilde{p}}}, \varphi^{\tilde{p}})$  is not robustly stable. Case \ref{case2} is symmetric.\footnote{Rigorously, in Case \ref{case2}, let $A=[p_*,p_*+\beta]$.
Since $\pi_0$ is atomless, $\pi_0(A)<\varepsilon$.
Construct $\varphi^A$ using Lemma \ref{lemma:perturbedfilter},
as in Case \ref{case1}.
For any prior $\pi$ feasible under this refinement and any
$p\in(p_*,p_*+\beta]$, agreement on $A$ and
Lemma \ref{lemma:equaltail} imply
\[
\pi([p,1])
 =\pi([p_*,1])-\pi([p_*,p)) 
 \geq \pi_0([p_*,1])-\pi_0([p_*,p)) 
 =\pi_0([p,1]).
\]
Hence $V^0(\mathcal{D}^{\tilde{p}},\pi)<V^0(\mathcal{D}^{p},\pi)$, so that $(\mathcal{D}^{\tilde{p}},\varphi^{\tilde{p}})$
is not robustly stable.
}

  \textbf{Step 4: } We shall show the ``if'' direction. If $\tilde{p}$ satisfies the two conditions in \Cref{t3}, then for any $p \in \supp \tilde{p}$, there exists $\varepsilon_p>0$ such that 
  \begin{equation} \label{b8}  0< {\theta} \pi_0([{\theta},1]) \leq p  \pi_0([p,1]), \text{ for all } \theta \in  (p-\varepsilon_p,p+\varepsilon_p),
  \end{equation} where the inequality ``$>0$'' is because $\pi_0$ has a full support. 
  Let $\underline{p}:= \min \supp \tilde{p}$ and $\bar{p}:=\max \supp \tilde{p}$. Then $0<\underline{p}\leq \bar{p}<1$. 
  
\textbf{Step 4.1: }Suppose for all $ p \in [\underline{p}, \bar{p}], $ we have
  \begin{equation} \label{eq:localoptimum1} p   \pi_0([p,1]) \leq V^0(\mathcal{D}^{\tilde{p}}, \pi_0)=\underline{p}   \pi_0([\underline{p},1])=\bar{p} \cdot \pi_0([\bar{p},1]), \end{equation}
where the last two equalities are by the first condition in \Cref{t3}.\footnote{It is worth mentioning that \eqref{eq:localoptimum1} includes the case where $\underline{p}=\bar{p}$.} 

Let \begin{equation}\label{eq:epsilonstar}
\varepsilon^*
:= \frac{1}{2}\,
\min\Big\{
    \pi_0([\underline{p}-\varepsilon_{\underline{p}},\,\underline{p}-\varepsilon_{\underline{p}}/2]),
    \ \pi_0([\bar{p},\,\bar{p}+\varepsilon_{\bar{p}}]),
    \ \underline{p}\,\pi_0([\underline{p},1]),
    \ \pi_0([0,\;\underline{p}\,\pi_0([\underline{p},1])])
\Big\}.
\end{equation}
It is readily verified that $\varepsilon^*>0$.\footnote{This follows from \eqref{b8} and the fact that $\pi_0$ has full support.}  Without loss of generality, it is valid to have $\varepsilon_{\bar{p}}$ and $\varepsilon_{\underline{p}}$ such that both \eqref{b8} and the following holds:
 \[
0<\underbrace{\underline{p}\,\pi_0([\underline{p},1])}_{\theta_1}
\;<\;
\underbrace{\underline{p}-\varepsilon_{\underline{p}}}_{\theta_2}
\;<\;
\underbrace{\underline{p}-\varepsilon_{\underline{p}}/2}_{\theta_3}
\;<\;
\underbrace{\bar{p}}_{\theta_4}
\;<\;
\underbrace{\bar{p}+\varepsilon_{\bar{p}}}_{\theta_5}<1.
\]
Consider any $\varphi^{\tilde{p}, \varepsilon^*}$ that reveals $\varepsilon^*$-truth about $\varphi^{\tilde{p}}$, i.e., there exists $A\subseteq \Theta$ with $\pi_0(A)<\varepsilon^*$ such that 
\begin{equation*}\Pi^{\mathcal{D}^{\tilde{p}}}(\varphi^{\tilde{p}, \varepsilon^*}, \varphi^{\tilde{p}, \varepsilon^*}_{\pi_0}) =\Pi^{\mathcal{D}^{\tilde{p}}}(\varphi^{\tilde{p}}, \varphi^{\tilde{p}}_{\pi_0}) \cap \{\pi \in \Delta(\Theta): \pi|_{A}=\pi_0|_A\}. 
\end{equation*}
 Let $\pi_{\varepsilon^*}\in \Pi^{\mathcal{D}^{\tilde{p}}}(\varphi^{\tilde{p}, \varepsilon^*}, \varphi_{\pi_0}^{\tilde{p}, \varepsilon^*})$ be such that the following conditions hold:
 \begin{enumerate}
 \item $\pi_{\varepsilon^*}$ agrees with $\pi_0$ on $A$ and $[\theta_3,\theta_4]$: \[\pi_{\varepsilon^*}|_{A \cup [\theta_3,\theta_4]}=\pi_0|_{A \cup [\theta_3,\theta_4]}.\]
\item $\pi_{\varepsilon^*}$ vanishes on $(\theta_1,\theta_3)$ and $(\theta_5,1]$ except on $A$: \[ \pi_{\varepsilon^*}|_{(\theta_1,\theta_3) \cup (\theta_5, 1] \backslash A}= 0.\]
\item $\pi_{\varepsilon^*}$ assigns at least as much measure as $\pi_0$ on $(\theta_4,\theta_5]$:
\begin{equation*}
    \pi_{\varepsilon^*}(E)  \ge  \pi_0(E) 
   \text{ for all } E \subseteq (\theta_4,\theta_5].
\end{equation*}
\item $\pi_{\varepsilon^*}$ shares with $\pi_0$ the same total measure on $[0,\theta_3]$ and on $[\theta_4,1]$, respectively:  
\begin{equation*}
\begin{aligned}
 &\pi_{\varepsilon^*}([0, \theta_3])=\pi_0([0,\theta_3]),\\
& \pi_{\varepsilon^*}([\theta_4,1]) = \pi_0([\theta_4,1]).
\end{aligned}
\end{equation*}
\end{enumerate}

We consider several cases. 
\begin{itemize}
    \item $ p \in [0, \theta_1]$: 
\[ p  \pi_{\varepsilon^*}([p,1]) \leq p \leq \theta_1 = \underline{p} \pi_0([\underline{p},1])=V^0(\mathcal{D}^{\tilde{p}}, \pi_0)=V^0(\mathcal{D}^{\tilde{p}}, \pi_{\varepsilon^*}),\] where the second to last equality is by the first condition in \Cref{t3} and the last equality is by \eqref{b5}.
    \item $p \in (\theta_1, \theta_2]$: \begin{align*}
p\,\pi_{\varepsilon^*}([p,1])
&\le
\theta_2\,\pi_{\varepsilon^*}((\theta_1,1])\\
&=
\theta_2\bigl[
\pi_{\varepsilon^*}((\theta_1,\theta_3))
+
\pi_{\varepsilon^*}([\theta_3,1])
\bigr]\\
&=
\theta_2\bigl[
\pi_{\varepsilon^*}(A\cap(\theta_1,\theta_3))
+
\pi_{\varepsilon^*}([\theta_3,1])
\bigr]\\
&\le
\theta_2\bigl[
\varepsilon^*
+
\pi_0([\theta_3,1])
\bigr]\\
&\le
\theta_2\bigl[
\pi_0([\theta_2,\theta_3])
+
\pi_0([\theta_3,1])
\bigr]\\
&=
\theta_2\,\pi_0([\theta_2,1])\\
&=
(\underline p-\varepsilon_{\underline p})\,
\pi_0([\underline p-\varepsilon_{\underline p},1])\\
&\le
V^0(\mathcal D^{\tilde p},\pi_{\varepsilon^*}).
\end{align*}
    \item $p \in (\theta_2, \theta_3)$: by \eqref{b8},  \[p  \pi_{\varepsilon^*}([p,1]) \leq p \cdot \pi_0([p,1]) \leq V^0(\mathcal{D}^{\tilde{p}}, \pi_{\varepsilon^*}).\]
    \item $p \in [\theta_3, \theta_4]$: by \eqref{b8} and \eqref{eq:localoptimum1},
\[p  \pi_{\varepsilon^*}([p,1]) = p \pi_0([p,1]) \leq V^0(\mathcal{D}^{\tilde{p}}, \pi_{\varepsilon^*}).\] 
    \item $ p \in (\theta_4, \theta_5]$: by \eqref{b8},
\begin{align*}
p\,\pi_{\varepsilon^*}([p,1])
&= p\bigl[\pi_{\varepsilon^*}([\theta_4,1]) - \pi_{\varepsilon^*}([\theta_4,p))\bigr] \\
 &\le  p\bigl[\pi_{0}([\theta_4,1]) - \pi_{0}([\theta_4,p))\bigr] \\
&\le p\,\pi_0([p,1]) \\
&\le V^0(\mathcal{D}^{\tilde{p}}, \pi_{\varepsilon^*}).
\end{align*}

    \item $p \in (\theta_5,1]$: by \eqref{eq:epsilonstar},
    \begin{equation*}
\label{ineq-display}
p\,\pi_{\varepsilon^*}([p,1])
\le
\pi_{\varepsilon^*}((\theta_5,1])
\le
\pi_{\varepsilon^*}(A)
\le
\varepsilon^*
\le
\underline{p}\,\pi_0([\underline{p},1])
=
V^0(\mathcal{D}^{\tilde{p}}, \pi_{\varepsilon^*}) .
\end{equation*}
\end{itemize}
Hence, $(\mathcal{D}^{\tilde{p}}, \varphi^{\tilde{p}, \varepsilon^*})$ is stable for any $\varphi^{\tilde{p}, \varepsilon^*}$ that reveals $\varepsilon^*$-truth about $\varphi^{\tilde{p}}$, which implies that $(\mathcal{D}^{\tilde{p}}, \varphi^{\tilde{p}})$ is robustly stable for any $\tilde{p}$ that satisfies the two conditions in \Cref{t3} and \eqref{eq:localoptimum1}. 

\textbf{Step 4.2: }Suppose \eqref{eq:localoptimum1} is not satisfied. Then by every price in the support of $\tilde{p}$ is a local maximizer of the revenue function, there exists a subinterval $[x, y] \subseteq [\underline{p}, \bar{p}]$ such that  
\begin{equation} \label{eq:uphill} p\pi_0([p,1]) \leq V^0(\mathcal{D}^{\tilde{p}}, \pi_{0}) \end{equation} 
for all $p\in [x, y]$. Furthermore, because $\pi_0$ is atomless, we can take the interval to be maximal such that no other subinterval $[\underline{p}, \bar{p}]\supseteq [x',y']\supsetneq [x,y]$ satisfies this property. Again, it follows from the fact that $\pi_0$ is atomless that
\begin{equation*}
x\pi_0([x,1])= y\pi_0([ y,1])=V^0(\mathcal{D}^{\tilde{p}}, \pi_{0}).
\end{equation*}
Because the revenue function doesn't oscillate infinite often, there can be only finitely many of such intervals.   Let  $\mathcal{I}=\{[x^k, y^k]\}_{k=1}^m$ be the collection of all such intervals; by definition, these intervals are pairwise disjoint, and $\supp \tilde{p}  \subseteq \bigcup_{k=1}^m [ x^k,y^k]$. %
For $k \in \{1,...,m\}$, if $\supp \tilde{p} \cap [x^k, y^k] $ is non-empty, define
\begin{equation*}
\begin{split}
    \underline{p}_k:= \min \supp \tilde{p} \cap [x^k,y^k], \\
    \bar{p}_k:= \max \supp \tilde{p} \cap [x^k,y^k],
    \end{split}
\end{equation*}
With some abuse of notation, we denote this collection as $\mathcal{J}=\{[ \underline{p}_k,\bar{p}_k]\}_{k=1}^l$. By definition, $\supp \tilde{p}  \subseteq \bigcup_{k=1}^l [ \underline{p}_k,\bar{p}_k]$ and
\begin{equation*}
   0< \underline{p}=\underline{p}_1 \leq \bar{p}_1<\underline{p}_2 \leq \bar{p}_2<...<\underline{p}_l \leq \bar{p}_l=\bar{p}<1.
\end{equation*}
The argument below replicate the argument for each of the following intervals $[0,\bar{p}_1],$ $[\bar{p}_1,\bar{p}_2]$,\dots, $[\bar{p}_{l-1},1]$

For any $j \in \{1,...,l\}$, let $\varepsilon_{\underline{p}_j}, \varepsilon_{\bar{p}_j}>0$ satisfy both \eqref{b8} and 
\begin{equation}
\label{eq:epsilonpj}
  \underline{p} \pi_0([\underline{p},1])<  \underline{p}_{1}-\varepsilon_{\underline{p}_{1}}<\bar{p}_{1}+\varepsilon_{\bar{p}_{1}}<\underline{p}_{2}-\varepsilon_{\underline{p}_{2}}<\bar{p}_{2}+\varepsilon_{\bar{p}_{2}}<\dots<\underline{p}_{l}-\varepsilon_{\underline{p}_{l}}<\bar{p}_{l}+\varepsilon_{\bar{p}_{l}}<1.
\end{equation}

Let 
\begin{equation}\label{eq:epsilondoublestar}
\varepsilon^{**}=\frac{1}{2}\min\!\left\{\min_{j=1,\dots,l}\pi_0([\underline{p}_j-\varepsilon_{\underline{p}_j},\,\underline{p}_j-\tfrac{\varepsilon_{\underline{p}_j}}{2}]),\;\min_{j=1,\dots,l}\pi_0([\bar{p}_j,\,\bar{p}_j+\varepsilon_{\bar{p}_j}]),\;\underline{p}\,\pi_0([\underline{p},1]),\;\pi_0([0,\,\underline{p}\,\pi_0([\underline{p},1])])\right\}.
\end{equation}

Consider any $\varphi^{\tilde{p}, \varepsilon^{**}}$ that reveals $\varepsilon^{**}$-truth about $\varphi^{\tilde{p}}$, i.e., there exists $A\subseteq \Theta$ with $\pi_0(A)<\varepsilon^{**}$ such that \begin{equation*}\Pi^{\mathcal{D}^{\tilde{p}}}(\varphi^{\tilde{p}, \varepsilon^{**}}, \varphi^{\tilde{p}, \varepsilon^{**}}_{\pi_0}) =\Pi^{\mathcal{D}^{\tilde{p}}}(\varphi^{\tilde{p}}, \varphi^{\tilde{p}}_{\pi_0}) \cap \{\pi \in \Delta(\Theta): \pi|_{A}=\pi_0|_A\}. \end{equation*}

Let $\pi_{\varepsilon^{**}} \in \Pi^{\mathcal{D}^{\tilde{p}}}(\varphi^{\tilde{p}, \varepsilon^{**}}, \varphi_{\pi_0}^{\tilde{p}, \varepsilon^{**}})$ be such that the following conditions hold:

\begin{enumerate}
    \item $\pi_{\varepsilon^{**}}$ agrees with $\pi_0$ on $A$ and $\bigcup_{j=1}^l [\underline{p}_j-\frac{\varepsilon_{\underline{p}_j}}{2}, \bar{p}_j]$: \[\pi_{\varepsilon^{**}}|_{A \bigcup \cup_{j=1}^l [\underline{p}_j-\frac{\varepsilon_{\underline{p}_j}}{2}, \bar{p}_j]}=\pi_0|_{A \bigcup \cup_{j=1}^l [\underline{p}_j-\frac{\varepsilon_{\underline{p}_j}}{2}, \bar{p}_j]}.\]
    \item $\pi_{\varepsilon^{**}}$ vanishes on $\left(
\bigcup_{j=1}^{l}
\left[
\bar{p}_{j-1}+\varepsilon_{\bar{p}_{j-1}},
\underline{p}_j-\frac{\varepsilon_{\underline{p}_j}}{2}
\right)
\right)
\cup
\left[\bar{p}_l+\varepsilon_{\bar{p}_l},1\right]$ except on $A$: \[ \pi_{\varepsilon^{**}}|_{ \left(
\bigcup_{j=1}^{l}
\left[
\bar{p}_{j-1}+\varepsilon_{\bar{p}_{j-1}},
\underline{p}_j-\frac{\varepsilon_{\underline{p}_j}}{2}
\right)
\right)
\cup
\left[\bar{p}_l+\varepsilon_{\bar{p}_l},1\right]}= 0,\]
    where $\bar{p}_0+\varepsilon_{\bar{p}_0}$ is understood to be $\underline{p}\pi_0([\underline{p},1])$, and each $[\bar{p}_{j-1}+\varepsilon_{\bar{p}_{j-1}}, \underline{p}_j-\frac{\varepsilon_{\underline{p}_j}}{2})$ is well-defined by \eqref{eq:epsilonpj}.
    \item For any $j \in \{1,...,l\}, \pi_{\varepsilon^{**}}$ assigns at least as much measure as $\pi_0$ on $(\bar{p}_j, \bar{p}_j+\varepsilon_{\bar{p}_j})$:
    \begin{equation*}
    \pi_{\varepsilon^{**}}(E)  \ge  \pi_0(E) 
   \text{ for all } E \subseteq (\bar{p}_j, \bar{p}_j+\varepsilon_{\bar{p}_j}), \forall j \in \{1,...,l\}.
\end{equation*}
\item For any $j \in \{1,...,l-1\}$, $\pi_{\varepsilon^{**}}$ shares with $\pi_0$ the same total measure on $[0, \underline{p}_1], [\bar{p}_l,1]$, and $[\bar{p}_j, \underline{p}_{j+1}]$:
\begin{equation*}
\begin{aligned}
 &\pi_{\varepsilon^{**}}([0, \underline{p}_1])=\pi_0([0,\underline{p}_1]),\\
& \pi_{\varepsilon^{**}}([\bar{p}_l,1]) = \pi_0([\bar{p}_l,1]), \\
& \pi_{\varepsilon^{**}}([\bar{p}_j, \underline{p}_{j+1}])=\pi_0([\bar{p}_j, \underline{p}_{j+1}]), \forall j \in \{1,...,l-1\}.
\end{aligned}
\end{equation*}
\end{enumerate}

We discuss several cases. 

\begin{itemize}
    \item $p \in [0, \underline{p} \cdot \pi_0([\underline{p},1])] $:
    \[ p  \pi_{\varepsilon^{**}}([p,1]) \leq p \leq  \underline{p} \pi_0([\underline{p},1])=V^0(\mathcal{D}^{\tilde{p}}, \pi_0)=V^0(\mathcal{D}^{\tilde{p}}, \pi_{\varepsilon^{**}}),\] where the second to last equality is by the first condition in \Cref{t3} and the last equality is by \eqref{b5}.
    \item For any $j \in \{1,...,l\}, p \in (\bar{p}_{j-1}+\varepsilon_{\bar{p}_{j-1}}, \underline{p}_j-\varepsilon_{\underline{p}_j}]$:
    \begin{align*}
p\,\pi_{\varepsilon^{**}}([p,1])
&\;\le\;
(\underline{p}_j-\varepsilon_{\underline{p}_j})\,\pi_{\varepsilon^{**}}((\bar{p}_{j-1}+\varepsilon_{\bar{p}_{j-1}},1])
\\ 
&\;=\;
(\underline{p}_j-\varepsilon_{\underline{p}_j})\!\left[\,
\pi_{\varepsilon^{**}}\bigl((\bar{p}_{j-1}+\varepsilon_{\bar{p}_{j-1}},\underline{p}_j-\frac{\varepsilon_{\underline{p}_j}}{2})\bigr)
\;+\;
\pi_{\varepsilon^{**}}([\underline{p}_j-\frac{\varepsilon_{\underline{p}_j}}{2},1])
\right]
\\ 
&\;=\;
(\underline{p}_j-\varepsilon_{\underline{p}_j})\!\left[\,
\pi_{\varepsilon^{**}}\bigl(A\cap(\bar{p}_{j-1}+\varepsilon_{\bar{p}_{j-1}},\underline{p}_j-\frac{\varepsilon_{\underline{p}_j}}{2})\bigr)
\;+\;
\pi_{\varepsilon^{**}}([\underline{p}_j-\frac{\varepsilon_{\underline{p}_j}}{2},1])
\right]
\\ 
 &\;\leq\;
(\underline{p}_j-\varepsilon_{\underline{p}_j})\!\left[\,
{\varepsilon^{**}}
\;+\;
\pi_{0}([\underline{p}_j-\frac{\varepsilon_{\underline{p}_j}}{2},1])
\right]
\\ 
 &\;\leq\;
(\underline{p}_j-\varepsilon_{\underline{p}_j})\!\left[\,
{\pi_0([\underline{p}_j-\varepsilon_{\underline{p}_j},\underline{p}_j-\frac{\varepsilon_{\underline{p}_j}}{2}])}
\;+\;
\pi_{0}([\underline{p}_j-\frac{\varepsilon_{\underline{p}_j}}{2},1])
\right]
\\[0.3em]
&\;=\;
(\underline{p}_j-\varepsilon_{\underline{p}_j})\,
\pi_{0}\bigl(
[\underline{p}_j-\varepsilon_{\underline{p}_j},1]
\bigr)
\\ 
&\;\le\;
V^{0}(\mathcal{D}^{\tilde{p}},\pi_{\varepsilon^{**}}),
\end{align*}
by \eqref{b8} and \eqref{eq:epsilondoublestar}.
    \item For any $j \in \{1,...,l\}, p \in (\underline{p}_j-\varepsilon_{\underline{p}_j}, \underline{p}_j-\dfrac{\varepsilon_{\underline{p}_j}}{2}]$: by \eqref{b8}, \[p \cdot \pi_{\varepsilon^{**}}([p,1]) \leq p \cdot \pi_0([p,1]) \leq V^0(\mathcal{D}^{\tilde{p}}, \pi_{\varepsilon^{**}}).\] 
    \item For any $j \in \{1,...,l\}, p \in [\underline{p}_j-\dfrac{\varepsilon_{\underline{p}_j}}{2}, \bar{p}_j]$: by \eqref{b8} and \eqref{eq:uphill}, \[ p \cdot \pi_{\varepsilon^{**}}([p,1])=p \cdot \pi_0([p,1]) \leq V^0(\mathcal{D}^{\tilde{p}}, \pi_{\varepsilon^{**}}). \] 
    \item For any $j \in \{1,...,l\}, p \in (\bar{p}_j, \bar{p}_j+\varepsilon_{\bar{p}_j}]$: by \eqref{b8}, %
\begin{align*}
p\,\pi_{\varepsilon^{**}}([p,1])
&= p\bigl[\pi_{\varepsilon^{**}}([\bar{p}_j,1]) - \pi_{\varepsilon^{**}}([\bar{p}_j,p))\bigr] \\
 &\le  p\bigl[\pi_{0}([\bar{p}_j,1]) - \pi_{0}([\bar{p}_j,p))\bigr] \\
&\le p\,\pi_0([p,1]), \\
&\le V^0(\mathcal{D}^{\tilde{p}}, \pi_{\varepsilon^{**}}).
\end{align*}
    \item $p \in (\bar{p}_l+\varepsilon_{\bar{p}_l}, 1]$: by \eqref{eq:epsilondoublestar}, \[p \cdot \pi_{\varepsilon^{**}}([p,1]) \leq \varepsilon^{**} \leq V^0(\mathcal{D}^{\tilde{p}}, \pi_{\varepsilon^{**}}). \] 
\end{itemize}
Hence, $(\mathcal{D}^{\tilde{p}}, \varphi^{\tilde{p}, \varepsilon^{**}})$ is stable for any $\varphi^{\tilde{p}, \varepsilon^{**}}$ that reveals $\varepsilon^{**}$-truth about $\varphi^{\tilde{p}}$.

In conclusion, if $\tilde{p}$ satisfies the two conditions in \Cref{t3}, then $(\mathcal{D}^{\tilde{p}}, \varphi^{\tilde{p}})$ is robustly stable, which is equivalent to $\tilde{p}$ is robustly stable. 
\end{proof}

\printbibliography[heading=bibintoc,title={References},resetnumbers=true]

\newpage
\pagenumbering{roman}
\section{Online Appendix}

\subsection{Kernel vs. Blackwell Informativeness Orders}\label{sec:Kernel-vs-Blackwell}
In this section, we prove the claim following \Cref{def:xinformative}, which we summarize below.
\begin{lemma}
    If $g$ is Blackwell more informative than $h$, then $g \succsim h$. 
\end{lemma}

\begin{proof}
Suppose $g : X \to \Delta(Y)$ and $h : X \to \Delta(Z)$ are stochastic mappings, and that $g$ is Blackwell more informative than $h$. By definition, there exists a stochastic mapping
\(
b : Y \to \Delta(Z)
\)
such that
\(
h = b \circ g,
\)
that is, for every $x \in X$ and every measurable set $E \subseteq Z$,
\[
h(x)(E) = \int_Y b(y)(E)\, g(x)(dy).
\]
For any prior $\mu \in \Delta(X)$, %
we obtain for any measurable $E \subseteq Z$,
\begin{align*}
h_\mu(E)
&= \int_X h(x)(E)\, \mu(dx) \\
&= \int_X \left( \int_Y b(y)(E)\, g(x)(dy) \right) \mu(dx) \\
&= \int_Y b(y)(E)\, g_\mu(dy),
\end{align*}
where the last equality follows from Fubini's theorem. Hence
\[
h_\mu = b \circ g_\mu.
\]
Fix $\mu_0 \in \Delta(X)$ and suppose $\mu \in \Delta(X)$ satisfies
\(
g_\mu = g_{\mu_0}.
\)
Then
\[
h_\mu
= b \circ g_\mu
= b \circ g_{\mu_0}
= h_{\mu_0}.
\]
Therefore,
\(
\{\mu \in \Delta(X) : g_\mu = g_{\mu_0}\}
\subseteq
\{\mu \in \Delta(X) : h_\mu = h_{\mu_0}\}.
\)
Thus $g \succsim h$. \end{proof}

\subsection{Incentive Compatibility and Incentive Implementability}\label{sec:incentive-implementability}
\Cref{ex:ic-converse} below shows that the reverse of \Cref{thm:stochastic-RP}(ii) does not hold: there exists an incentive compatible stochastic fictitious mechanism that cannot be implemented by an incentive compatible augmented mechanism.
\begin{example}\label{ex:ic-converse} 
    Consider a single
agent with $\Theta=O=[0,3]$ and payoff function satisfying
\(u(o,\theta)=1\) if \(o=\theta\), and \(0\) otherwise.
Define
\[
\varphi(\theta)=\delta(\theta)=
\begin{cases}
\frac12\mathbf{1}_0+\frac12\mathbf{1}_1, & \theta=1,\\[3pt]
\frac12\mathbf{1}_0+\frac12\mathbf{1}_2, & \theta=2,\\[3pt]
\frac12\mathbf{1}_1+\frac12\mathbf{1}_2, & \theta=3,\\[3pt]
\mathbf{1}_0, & \text{otherwise}.
\end{cases}
\]
By construction, $\varphi$ is a filter of the direct mechanism $\mathcal D =(\Theta, \delta)$.
Moreover, $\mathcal D$ is incentive compatible:
types $0,1,2$ obtain their maximal attainable payoffs
$1,\frac12,\frac12$, respectively, while every other
type obtains zero from every report.

Consider the two priors
\[
\mu=\frac12\mathbf{1}_0+\frac12\mathbf{1}_3
\quad \textup{and} \quad
\nu=\frac12\mathbf{1}_1+\frac12\mathbf{1}_2.
\]
They are indistinguishable under $\varphi$ because
\[
\varphi_\mu=\varphi_\nu
=\frac12\mathbf{1}_0
+\frac14\mathbf{1}_1+\frac14\mathbf{1}_2.
\]
Suppose that the fictitious direct mechanism $(\mathcal D,\varphi)$ is implemented by
an incentive compatible augmented mechanism
$(\mathcal M,\sigma)=(M,\omega,\sigma)$.
Let
\[
E:=\{m\in M:\omega(m)(\{0\})=1\}.
\]
Outcome equivalence at type $0$ requires
$\sigma(0)(E)=1$.
For each $\theta\in\{1,2\}$, outcome equivalence implies that
type $\theta$ obtains a payoff of $\frac12$, and incentive compatibility
requires $\omega(m)(\{\theta\})\leq\frac12$ for every $m\in M$.
Since messages in $E$ give type $\theta$ a payoff of zero,
$\sigma(\theta)(E)=0$.
Consequently,
\[
\sigma_\mu(E)\geq\frac12,
\qquad
\sigma_\nu(E)=0.
\]
Thus $\sigma_\mu\neq\sigma_\nu$ despite
$\varphi_\mu=\varphi_\nu$, contradicting the requirement of information equivalence
$\sigma\sim\varphi$.\qed
\end{example}

\subsection{Proof of \Cref{lemma:perturbedfilter}}\label{sec:lemma:perturbedfilter}
\begin{proof} The proof of part~\textup{(i)} proceeds in three steps.

    \textbf{Step 1: }  Consider any product set
$A=\prod_{i\in N}A^i\subseteq\Theta$
with each $A^i$ Borel and a fictitious direct mechanism $(\mathcal{D}, \varphi)$. For each $i$, choose
$\bar\theta^i\in\Theta^i\setminus A^i$
when $A^i\ne\Theta^i$, and choose
$\bar\theta^i\in\Theta^i$ arbitrarily
when $A^i=\Theta^i$.
Define $\sigma^i: \Theta^i \rightarrow \Delta(\Theta^i \times \Theta^i)$  by
    \begin{equation}\label{eqn:def_sigma}
    \sigma^i(\theta^i):=
    \begin{cases}
        {1}_{\bar{\theta}^i} \otimes  {1_{\theta^i}}, & \text{if } \theta^i \in A^i, \\
        \varphi^i(\theta^i) \otimes {1_{\bar{\theta}^i}}, & \text{otherwise,} 
    \end{cases}
    \end{equation}
where ${1_{\theta^i}}$ and $1_{\bar{\theta}^i}$ are Dirac measures, and $\otimes$ denotes the measure induced by product measures. Let $\sigma(\theta) =(\sigma^i(\theta^i))_{i\in N}$.

If $A^i=\Theta^i$, then
\(
\sigma^i_{\pi^i}
=\bm{1}_{\bar\theta^i}\otimes\pi^i.
\)
Hence $\sigma^i_{\pi^i}=\sigma^i_{\pi_0^i}$
if and only if $\pi^i=\pi_0^i$, so
\eqref{eq:sigmai} holds immediately.
In what follows, we therefore need only consider
coordinates with $A^i\ne\Theta^i$.

We shall show the existence of $\omega$ such that $\omega \circ \sigma=\delta$. Since $\varphi$ is a filter of $\mathcal{D}$, there exists a stochastic mapping $g:\Theta\to\Delta(O)$ such that $\delta=g\circ\varphi$. For each $i$, define $b^i:\Theta^i\times\Theta^i\to\Delta(\Theta^i)$ by
\begin{equation*}\label{eqn:def_b}
b^i(\tilde{\theta}^i,\hat{\theta}^i):=
\begin{cases}
\varphi^i(\hat{\theta}^i), & \text{if } \hat{\theta}^i\in A^i,\\
1_{\tilde{\theta}^i}, & \text{otherwise}.
\end{cases}
\end{equation*}
Then $b^i$ is a stochastic mapping.\footnote{For every measurable $E\subseteq\Theta^i$, $b^i(\tilde{\theta}^i,\hat{\theta}^i)(E)=1_{\hat{\theta}^i}(A^i)\,\varphi^i(\hat{\theta}^i)(E)+1_{\hat{\theta}^i}(\Theta^i\setminus A^i)\,1_{\tilde{\theta}^i}(E)$, which is measurable in $(\tilde{\theta}^i,\hat{\theta}^i)$.} Let $b=(b^1,\ldots,b^n)$ be the joint function, and define the outcome function $\omega:\Theta\times\Theta\to\Delta(O)$ by
\begin{equation*}\label{eqn:def_omega}
\omega:=g\circ b.
\end{equation*}

We claim that $b^i\circ\sigma^i=\varphi^i$ for every $i$. If $\theta^i\in A^i$, then by \eqref{eqn:def_sigma}, $(b^i\circ\sigma^i)(\theta^i)=b^i(\bar{\theta}^i,\theta^i)=\varphi^i(\theta^i)$. If $\theta^i\notin A^i$, then, since $\bar{\theta}^i\notin A^i$, for every measurable $E\subseteq\Theta^i$,
\[
(b^i\circ\sigma^i)(\theta^i)(E)=\int_{\Theta^i}b^i(\tilde{\theta}^i,\bar{\theta}^i)(E)\,\varphi^i(\theta^i)(d\tilde{\theta}^i)=\int_{\Theta^i}1_{\tilde{\theta}^i}(E)\,\varphi^i(\theta^i)(d\tilde{\theta}^i)=\varphi^i(\theta^i)(E).
\]
Since $\sigma$ and $b$ are joint functions, $(b\circ\sigma)(\theta)=\bigl((b^i\circ\sigma^i)(\theta^i)\bigr)_{i\in N}=\varphi(\theta)$ for every $\theta\in\Theta$. Hence,
\[
\omega\circ\sigma=(g\circ b)\circ\sigma=g\circ(b\circ\sigma)=g\circ\varphi=\delta.
\]

We now proceed to show that \eqref{eq:sigmai} holds for $\sigma^i$ constructed by \eqref{eqn:def_sigma}. 

\textbf{Step 2: } We show the ``$\supseteq$'' direction of \eqref{eq:sigmai}. 

For any $\pi^i \in \Delta(\Theta^i)$ belongs to the right-hand side of \eqref{eq:sigmai}, i.e., $\varphi^i_{\pi^i}=\varphi^i_{\pi_0^i}$ and $\pi^i|_{A^i}=\pi_0^i|_{A^i}$, and for any measurable ${E^i} \subseteq \Theta^i \times \Theta^i$, it follows from \eqref{eqn:def_sigma} that \begin{equation}\label{eqn:sigma_pi}
    \begin{split}
        \sigma^i_{\pi^i}\!({E}^i)
        &=\displaystyle{\int}_{\Theta^i} \sigma^i(\theta^i)\!({E}^i)\,\pi^i(d\theta^i) \\
        &= \displaystyle{\int}_{A^i} \sigma^i(\theta^i)\!({E}^i)\,\pi^i(d\theta^i)
        +\displaystyle{\int}_{\Theta^i \backslash A^i} \sigma^i(\theta^i)\!({E}^i)\,\pi^i(d\theta^i) \\
        &= \pi^i(E_1^i)+ \displaystyle{\int}_{\Theta^i \backslash A^i} \varphi^i(\theta^i)\!(E_2^i)\,\pi^i(d\theta^i),
    \end{split}
\end{equation}
where\[
\begin{array}{l}
E_1^i = \{\theta^i \in A^i : (\bar{\theta}^i, \theta^i) \in {E}^i\},\\[6pt]
E_2^i = \{\theta^i \in \Theta^i : (\theta^i, \bar{\theta}^i) \in {E}^i\}.
\end{array}
\]
It follows from the fact that $E_1^i$ is a measurable subset of $A^i$ and  $\pi^i|_{A^i}=\pi_0^i|_{A^i}$ that
\begin{equation}\label{eqn:pi_equ}
\pi^i(E_1^i)=\pi_0^i(E_1^i).
\end{equation}
In addition, since $E_2^i$ is a measurable subset of $\Theta^i$, it follows from  $\varphi^i_{\pi^i}=\varphi^i_{\pi_0^i}$ and $\pi^i|_{A^i}=\pi_0^i|_{A^i}$ that \begin{equation*}
\begin{aligned}
\int_{\Theta^i} \varphi^i(\theta^i)\!(E_2^i)\,\pi^i(d\theta^i)
&= \int_{\Theta^i} \varphi^i(\theta^i)\!(E_2^i)\,\pi_0^i(d\theta^i),\\[6pt]
\int_{A^i} \varphi^i(\theta^i)\!(E_2^i)\,\pi^i(d\theta^i)
&= \int_{A^i} \varphi^i(\theta^i)\!(E_2^i)\,\pi_0^i(d\theta^i).
\end{aligned}
\end{equation*}
Therefore,  \begin{equation}\label{eqn:varphi_A_complement}
\displaystyle{\int}_{\Theta^i \backslash A^i} \varphi^i(\theta^i)\!(E_2^i)\,\pi^i(d\theta^i)
= \displaystyle{\int}_{\Theta^i \backslash A^i} \varphi^i(\theta^i)\!(E_2^i)\,\pi_0^i(d\theta^i).
\end{equation}
Hence, by \eqref{eqn:sigma_pi}, \eqref{eqn:pi_equ}, and \eqref{eqn:varphi_A_complement},
\begin{equation*}
\sigma^i_{\pi^i}\!({E}^i)
=\pi_0^i(E_1^i)+\displaystyle{\int}_{\Theta^i \backslash A^i} \varphi^i(\theta^i)\!(E_2^i)\,\pi_0^i(d\theta^i)
=\sigma^i_{\pi_0^i}\!(E^i),
\end{equation*}
which implies that $\sigma^i_{\pi^i}=\sigma^i_{\pi_0^i}$. Therefore, $\pi^i$ belongs to the left-hand side of \eqref{eq:sigmai}.

\textbf{Step 3: } We now show the ``$\subseteq$'' direction of \eqref{eq:sigmai}. 

Consider any $\pi^i$ such that $\sigma^i_{\pi^i}=\sigma^i_{\pi_0^i}$. For any measurable $F^i \subseteq A^i \not\ni \bar{\theta}^{i}$, we have 
\begin{equation*}
\begin{aligned}
\{\theta^i \in A^i : (\bar{\theta}^{i}, \theta^i) \in \{\bar{\theta}^{i}\} \times F^i\} &= F^i,\\[6pt]
\{\theta^i \in \Theta^i : (\theta^i, \bar{\theta}^{i}) \in \{\bar{\theta}^{i}\} \times F^i\} &= \emptyset.
\end{aligned}
\end{equation*}
Therefore, by \eqref{eqn:def_sigma}, 
\begin{equation*}
    \pi^i(F^i)=\sigma^i_{\pi^i}(\{\bar{\theta}^{i}\} \times F^i)=\sigma^i_{\pi_0^i}(\{\bar{\theta}^{i}\} \times F^i)=\pi_0^i(F^i).
\end{equation*}
Therefore, $\pi^i|_{A^i}=\pi_0^i|_{A^i}$.  

For any measurable $\tilde{F}^i \subseteq \Theta^i$,  we have 
\begin{equation*}
\begin{aligned}
\{\theta^i \in A^i : (\bar{\theta}^i, \theta^i) \in \tilde{F}^i \times \{\bar{\theta}^i\}\} &= \emptyset,\\[6pt]
\{\theta^i \in \Theta^i : (\theta^i, \bar{\theta}^i) \in \tilde{F}^i \times \{\bar{\theta}^i\}\} &= \tilde{F}^i.
\end{aligned}
\end{equation*}
Therefore,
\begin{equation*} \label{b.1}
\begin{split}
\sigma^i_{\pi^i}\!(\tilde{F}^i \times \{\bar{\theta}^i\})
&= \displaystyle{\int}_{\Theta^i } \sigma^i(\theta^i)\!(\tilde{F}^i\times \{\bar{\theta}^i\})\,\pi^i(d\theta^i)\\
&= \displaystyle{\int}_{\Theta^i \setminus A^i } \sigma^i(\theta^i)\!(\tilde{F}^i\times \{\bar{\theta}^i\})\,\pi^i(d\theta^i)
+ \displaystyle{\int}_{ A^i } \sigma^i(\theta^i)\!(\tilde{F}^i\times \{\bar{\theta}^i\})\,\pi^i(d\theta^i)\\
&=\displaystyle{\int}_{\Theta^i \backslash A^i} \varphi^i(\theta^i)\!(\tilde{F}^i)\,\pi^i(d\theta^i) .
\end{split}
\end{equation*}
Similarly,
\[
\sigma^i_{\pi_0^i}\!(\tilde{F}^i \times \{\bar{\theta}^i\})
=\displaystyle{\int}_{\Theta^i \backslash A^i} \varphi^i(\theta^i)\!(\tilde{F}^i)\,\pi_0^i(d\theta^i).
\]
Since $\sigma^i_{\pi^i}\!(\tilde{F}^i \times \{\bar{\theta}^i\})=\sigma^i_{\pi_0^i}\!(\tilde{F}^i \times \{\bar{\theta}^i\})$, the two displays above imply that  
\begin{equation*}
    \displaystyle{\int}_{\Theta^i \backslash A^i} \varphi^i(\theta^i)\!(\tilde{F}^i)\,\pi^i(d\theta^i)
    =\displaystyle{\int}_{\Theta^i \backslash A^i} \varphi^i(\theta^i)\!(\tilde{F}^i)\,\pi_0^i(d\theta^i).
\end{equation*}

Since $\pi^i|_{A^i}=\pi_0^i|_{A^i}$, we have
\begin{equation*}
\displaystyle{\int}_{A^i} \varphi^i(\theta^i)\!(\tilde{F}^i)\,\pi^i(d\theta^i)
=\displaystyle{\int}_{A^i} \varphi^i(\theta^i)\!(\tilde{F}^i)\,\pi_0^i(d\theta^i).
\end{equation*} 
Therefore,
\begin{equation*}
\displaystyle{\int}_{\Theta^i} \varphi^i(\theta^i)\!(\tilde{F}^i)\,\pi^i(d\theta^i)
=\displaystyle{\int}_{\Theta^i} \varphi^i(\theta^i)\!(\tilde{F}^i)\,\pi_0^i(d\theta^i).
\end{equation*}
Hence, $\varphi^i_{\pi^i}=\varphi^i_{\pi_0^i}$. Therefore, $\pi^i$ belongs to the right-hand side of \eqref{eq:sigmai}.

\medskip

We now prove part~\textup{(ii)}. Choose the fictitious direct mechanism representation above
by fixing, for each $i$, a Borel isomorphism
$h^i:\Theta^i\times\Theta^i\to\Theta^i$
and setting
\[
\varphi^{A^i}:=h^i\circ\sigma^i.
\]
Writing $h=(h^i)_{i\in N}$, we have
\[
\varphi^A=h\circ\sigma,
\qquad
\delta=(\omega\circ h^{-1})\circ\varphi^A.
\]
Thus $\varphi^A$ is a filter and
$\varphi^A\sim\sigma$, as required in \textup{(i)}.
If $\varphi$ is deterministic, then $\sigma$ is pure
by \eqref{eqn:def_sigma}, so $\varphi^A$ is deterministic.
This proves \textup{(ii)}.

We now prove part~\textup{(iii)}. Suppose that
$(\mathcal D,\varphi)$ is implemented by an incentive
compatible augmented mechanism $(\mathcal M_0,\tau)$,
where $\mathcal M_0=(M_0,\omega_0)$.
After a Borel relabeling of messages, we may take
$M_0^i=\Theta^i$ for each $i$.

We now apply the construction in Step~1 with $\tau^i$
in place of $\varphi^i$ to define
\[
\widetilde\sigma^i(\theta^i):=
\begin{cases}
\bm{1}_{\bar\theta^i}\otimes\bm{1}_{\theta^i},
& \theta^i\in A^i,\\
\tau^i(\theta^i)\otimes\bm{1}_{\bar\theta^i},
& \theta^i\notin A^i,
\end{cases}
\]
and
\[
\widetilde b^i(\tilde\theta^i,\hat\theta^i):=
\begin{cases}
\tau^i(\hat\theta^i), & \hat\theta^i\in A^i,\\
\bm{1}_{\tilde\theta^i}, & \hat\theta^i\notin A^i.
\end{cases}
\]
Similarly, define
\[
\widetilde\omega:=\omega_0\circ\widetilde b.
\]
Then
\[
\widetilde b\circ\widetilde\sigma=\tau,
\qquad
\widetilde\omega\circ\widetilde\sigma
=\omega_0\circ\tau=\delta.
\]
By construction, $\widetilde b^i$
maps each message in the new mechanism to a lottery
over messages in $\mathcal M_0$.
Thus every unilateral deviation in the new mechanism
can be reproduced by a mixed-message deviation
in $\mathcal M_0$, while the prescribed strategies
reproduce $\tau$.
The incentive compatibility of $(\mathcal M_0,\tau)$
therefore implies incentive compatibility of the new
augmented mechanism in the same sense.

The calculations in Steps~2--3 apply to any pair of
priors. Hence, for every $i$ and
$\mu,\nu\in\Delta(\Theta^i)$,
\[
\begin{aligned}
\widetilde\sigma^i_\mu
=\widetilde\sigma^i_\nu
&\iff
\tau^i_\mu=\tau^i_\nu
\ \text{and}\ \mu|_{A^i}=\nu|_{A^i}\\
&\iff
\varphi^i_\mu=\varphi^i_\nu
\ \text{and}\ \mu|_{A^i}=\nu|_{A^i}\\
&\iff
\sigma^i_\mu=\sigma^i_\nu,
\end{aligned}
\]
where the second equivalence uses
$\tau^i\sim\varphi^i$.
Thus $\widetilde\sigma^i\sim\sigma^i$ for every $i$,
and \Cref{lemma:independent} gives
\[
\widetilde\sigma\sim\sigma\sim\varphi^A.
\]
Therefore, the incentive compatible augmented mechanism
with outcome function $\widetilde\omega$ and strategies
$\widetilde\sigma$ implements $(\mathcal D,\varphi^A)$.
This proves \textup{(iii)}.
\end{proof}
\end{document}